\documentclass[journal]{IEEEtran}

\usepackage[utf8]{inputenc} 
\usepackage[T1]{fontenc}
\usepackage{url}              

\usepackage[noadjust]{cite}             

\usepackage[cmex10]{amsmath}  
\usepackage{mleftright}       
\mleftright                   

\usepackage{graphicx}         
\usepackage{booktabs}         

\usepackage[caption=false,font=footnotesize]{subfig}
                            
\usepackage{amssymb}
\usepackage[overload]{empheq}
\usepackage{xcolor}
\usepackage{mathtools}
\usepackage{soul}

\newcommand{\defeq}{\vcentcolon=}

\newtheorem{Definition}{Definition}
\newtheorem{Lemma}{\bf Lemma}
\newtheorem{Corollary}{\bf Corollary}
\newtheorem{Proposition}[Lemma]{\bf Proposition}
\newtheorem{Theorem}{\bf Theorem}

\newtheorem{Remark}{Remark}

\def\Pr{{\rm \mathbf {Pr}}}
\def\E{{\rm \mathbf  E}}

\newcommand{\Sample}{\mathrm{Sample}}

\newcommand{\eps}{\varepsilon}
\newcommand{\X}{\mathcal{X}}
\newcommand{\Y}{\mathcal{Y}}
\newcommand{\Q}{\mathcal{Q}}

\newcommand{\I}{\mathrm{I}}
\newcommand{\II}{\mathrm{II}}

\newcommand{\opt}{\mathrm{OPT}}

\usepackage{geometry}
 \allowdisplaybreaks

 \makeatletter
\newcommand\blfootnote[1]{%
  \begingroup
  \renewcommand{\@makefntext}[1]{\noindent\makebox[1.8em][r]#1}
  \renewcommand\thefootnote{}\footnote{#1}%
  \addtocounter{footnote}{-1}%
  \endgroup
}
\makeatother

\begin{document}

\title{Total Variation Meets Differential Privacy} 
\author{%
  \IEEEauthorblockN{Elena Ghazi$^\ast$, Ibrahim Issa$^\ast$$^\dagger$ \\}
  \IEEEauthorblockA{%
    $^\ast$Electrical and Computer Engineering Department, American University of Beirut, Beirut, Lebanon\\
    $^\dagger$Center for Advanced Mathematical Sciences, American University of Beirut, Beirut, Lebanon \\
    ebg03@mail.aub.edu, ibrahim.issa@aub.edu.lb} 
}

\maketitle

\begin{abstract}

The framework of approximate differential privacy is considered, and augmented by leveraging the notion of ``the total variation of a (privacy-preserving) mechanism'' (denoted by $\eta$-TV). With this refinement, an exact composition result is derived, and shown to be significantly tighter than the optimal bounds for differential privacy (which do not consider the total variation). Furthermore, it is shown that $(\eps,\delta)$-DP with $\eta$-TV is closed under subsampling. The induced total variation of commonly used mechanisms are computed. Moreover, the notion of total variation of a mechanism is studied in the local privacy setting and privacy-utility tradeoffs are investigated. In particular, total variation distance and KL divergence are considered as utility functions and studied through the lens of contraction coefficients. Finally, the results are compared and connected to the locally differentially private setting.

\blfootnote{This work was presented in part at the IEEE International Symposium on Information Theory, June 2023 (Taipei, Taiwan). Part of this work was done while Ibrahim Issa was a visiting professor at EPFL (Lausanne, Switzerland). Elena Ghazi was at the American University of Beirut (Beirut, Lebanon). This work was supported by the University Research Board at the American University of Beirut.}
\end{abstract}


\section{Introduction}

Given a database $D$ that contains private records (e.g., medical records), one would like to design a mechanism $M$ to answer queries issued by a curious (but not necessarily malicious) analyst (e.g., medical researchers). Such mechanism should provide useful answers while preserving the privacy of the individuals whose records are in $D$. In this context, differential privacy (DP)~\cite{DworkDP} was developed, and subsequently widely adopted, to quantify the privacy guarantees of a given mechanism $M$. Essentially, given any database $D$, any possible record, and the output $M(D)$ of a differentially-private $M$, the analyst cannot tell (except with insignificant probability) whether or not the record appears in $D$.

However, the analyst may issue several queries, thus degrading the privacy guarantees. Bounds on the resulting guarantees are called \emph{composition theorems}. Notably, pure differential privacy yields pessimistic composition theorems. As such, researchers have developed relaxations and variations of differential privacy to yield better composition behaviour, e.g., approximate DP~\cite{ApproxDiff2006}, R\'enyi-DP~\cite{mironov2017renyi}, concentrated-DP~\cite{dwork2016concentrated}, Gaussian DP~\cite{GaussianDP}, etc.

Herein, we focus on approximate DP and Gaussian DP as they have the clearest operational interpretation. In particular, Wasserman and Zhou~\cite{DiffPrivacyAsHT} and Kairouz \emph{et al.}~\cite{TheCompositionTheoremForDP} provided an equivalent characterization of $(\eps,\delta)$-DP in terms of the achievable region of errors (type I and type II) of a binary hypothesis testing experiment.  Utilizing this view, Kairouz \emph{et al.}~\cite{TheCompositionTheoremForDP} proved an exact composition result for $(\eps,\delta)$-DP in terms of a family of $\{(\eps_j,\delta_j)\}_j$ parameters. As observed by Dong \emph{et al.}~\cite{GaussianDP}, the binary hypothesis view also elucidates why composition results for $(\eps,\delta)$-DP are pessimistic: the framework is not rich enough to capture the induced privacy region of the composed mechanisms (with one $(\eps,\delta)$ pair). Instead, they parameterized the privacy guarantee by (the lower convex envelope of) the privacy region itself, and denoted it by $f$-DP~\cite{GaussianDP}. Hence, they showed that $f$-DP is closed under composition, and proved an interesting limiting behavior wherein, under certain assumptions, the limit of the composition of $n$ $f$-DP mechanisms converges to the guarantees of the Gaussian mechanism. 

In this work, we introduce a simple refinement to (approximate) differential privacy that yields better composition results. Namely, we leverage the total variation (TV) of a mechanism, denoted by $\eta$-TV, so that we keep track of both the $(\eps,\delta)$-parameters for DP and the $\eta$-parameter for TV. This retains an equivalent characterization in terms of binary hypothesis testing. In addition to enabling tighter analyses, explicitly incorporating total variation offers important advantages in the context of learning algorithms. In particular, $\eta$-TV has emerged as the target measure to control \emph{membership inference attacks} (MIA) --- which pose an orthogonal concern compared to differential privacy. Moreover, $\eta$-TV can be used to bound the generalization error of the algorithms.

Herein, our contributions consist of the following:
\begin{itemize}
    \item We prove an \emph{exact} composition bound for $(\eps,\delta)$-DP coupled with $\eta$-TV (Section~\ref{sec:main}).
    \item We show the provided bounds can be significantly tighter than bounds where total variation is not taken into consideration (using examples, as well as  asymptotic analysis~\cite{GaussianDP}) (Sections~\ref{sec:main} and~\ref{sec:limit}).
    \item We show that $(\eps,\delta)$-DP with $\eta$-TV is closed under subsampling (where the mechanism computes the query answer on a random subset of the database) (Section~\ref{sec:subsampling}).
    \item We compute the total variation of commonly used mechanisms and demonstrate an interesting connection with the staircase mechanism~\cite{StaircaseMechanism} (Section~\ref{sec:app}). 
    \item We analyze the differentially private stochastic gradient descent algorithm (Section~\ref{sec:SGD}).
\end{itemize}

We also study the privacy setting in which data remains private even from the statistician. Given distributions $P_0$ and $P_1$ and an $\eps$-locally differentially private ($\eps$-LDP) privatization mechanism $Q$, let $M_0$ and $M_1$ be the induced output marginals. Herein, the goal is to design a mechanism that satisfies a certain privacy constraint (to satisfy the data owners), while maximizing a certain statistical utility function (to satisfy the statistician).
Duchi \emph{et al.} \cite{duchi2014local} proved bounds on the symmetrized KL divergence between $M_0$ and $M_1$ in the $\eps$-LDP setting. Kairouz \emph{et al.} \cite{ExtremalMechanisms} provided an $\eps$-LDP binary mechanism that maximizes the total variation between $M_0$ and $M_1$ for all $P_0$ and $P_1$, and approximates the mechanism that maximizes the KL divergence between the marginals. We consider the notion of total variation of a mechanism in the local privacy setting and study privacy-utility tradeoffs (Section~\ref{sec:localDP}). In particular, our contributions consist of the following.

    \begin{itemize}
        \item We provide a mechanism that maximizes the total variation between the marginals in the case of $\eps$-LDP with $\eta$-TV.
        \item We generalize bounds on the contraction coefficient of KL divergence for a privacy-preserving mechanism ~\cite{duchi2014local,ExtremalMechanisms,asoodeh:22} by accounting simultaneously for $\eps$-LDP and $\eta$-TV constraints.
        \item We generalize bounds on $\chi^2$ of the marginal distributions in terms of the total variation of the input distributions~\cite{duchi2014local,asoodeh:22}.  
        \item Finally, given an $\eps$-LDP mechanism, we show how to construct an $\eps$-LDP with $\eta$-TV mechanism, and use this result to connect the corresponding privacy-utility tradeoffs in the two settings (i.e., enforcing $\eps$-LDP constraint versus enforcing $\eps$-LDP and an $\eta$-TV constraint).
    \end{itemize}
      
Finally, it is worth mentioning that it is typically simple to compute bounds on total variation using, for instance, KL divergence or Chernoff information. As such, the computational overhead of keeping track of total variation is relatively low. Furthermore, it may be estimated from data in a ``black-box'' manner~\cite{BayesSecurity}.

{\bf{Prior Work:}} Total variation is arguably a ``natural'' measure to consider in privacy analysis and has indeed  been used in the literature as a privacy metric: Barber and Duchi \cite{barber2014privacy} compared various definitions of privacy for several estimation problems, including \emph{$\alpha$-total variation privacy}. Geng \emph{et al.} \cite{Geng_OptimalNoiseAddingMechanism} derived the optimal $(0,\delta)$-differentially private noise-adding mechanism for single real-valued query function under a cost-minimization framework. Jia \emph{et al.} \cite{TVD_MIA} used the notion of \emph{total variation distance privacy} to accurately estimate privacy risk, despite it being a weaker privacy definition than differential privacy. The total variation also appeared as a bound on generalization metrics in the context of machine learning, and consequently as a bound on vulnerability to Membership Inference Attacks (MIAs). Dwork \emph{et al.} \cite{fairness-through-awareness} examined the total variation as a notion of fairness. Bassily \emph{et al.}~\cite{AlgorithmicStability} defined \emph{TV-stability} as a notion of algorithmic stability used to bound generalization error, which later appeared in Raginsky \emph{et al.}'s analysis of bias in learning algorithms~\cite{Analysis-of-Stability-Raginsky}. Kulynych \emph{et al.} \cite{DisparateVulnerability,WYSIWYG} proved through the notion of \emph{distributional generalization} that the total variation bounds vulnerability to MIAs, as well as disparate vulnerability against MIAs (unequal success rate of MIAs against different population subgroups). Chatzikokolakis \emph{et al.} \cite{BayesSecurity} also studied \emph{Bayes security}, a security metric inspired by the cryptographic advantage, equal to the complement of the total variation. Additionally (as will be seen in this paper), it naturally appears in many existing analysis. Note that our work differs from existing frameworks that leverage the notion of total variation in a privacy setting (e.g., \cite{barber2014privacy,Geng_OptimalNoiseAddingMechanism,TVD_MIA}) in that we suggest keeping track of the total variation of a mechanism \emph{in addition to} differential privacy parameters, rather than as a standalone metric. This framework continues to leverage the strong guarantees of differential privacy, while additionally accounting for the mechanism-specific property of total variation.

{\bf{Novel Contributions.}} In addition to our ISIT findings, we consider applications of our results in the context of membership inference attacks, as well as differentially private learning algorithms (namely, noisy stochastic gradient descent). Furthermore, we study the notion of total variation of a mechanism in the local privacy setting. We obtain generalized privacy-utility bounds, showing that one can obtain more precise guarantees by accounting for both $\eps$ and $\eta$ parameters.


\section{Preliminaries and Definitions}

Fix an alphabet $\X$ and an integer $m \in \mathbb{N}$. A database $D$ is an element in $\X^m$. Two databases that differ in one entry are called neighboring databases. A (query-answering) mechanism $M$ is a randomized map from $\X^m$ to an output space, which we denote by $\Y$. 

\begin{Definition}[Differential Privacy (DP)~\cite{DworkDP}]
Given $\eps \geq 0$ and $\delta \in [0,1]$, a mechanism $M$ is $(\eps,\delta)$-differentially private if, for all neighboring databases $D_0$ and $D_1$ and all $S \subseteq \Y$,
\begin{align*}
    \Pr( M(D_0) \in S) \leq e^\eps \Pr(M(D_1) \in S) + \delta. 
\end{align*}
\end{Definition}
Wasserman and Zhou~\cite{DiffPrivacyAsHT} and Kairouz \emph{et al.}~\cite{TheCompositionTheoremForDP} provide an equivalent characterization of $(\eps,0)$- and $(\eps,\delta)$-DP, respectively, in terms of binary hypothesis testing. In particular, consider a mechanism $M$, two neighboring databases $D_0$ and $D_1$, and let $P_0$ and $P_1$ be the corresponding distributions over $\Y$ of $M(D_0)$ and $M(D_1)$.  Given a random output $Y \in \Y$, an adversary aims to distinguish between 
\begin{center}
    $H_0$: $Y \sim P_0$, \quad and \quad $H_1 $: $Y \sim P_1$.
\end{center}

Let $h$ be any (possibly randomized) decision function, $h: \Y \rightarrow \{0,1\}$, and denote by $\beta_\I(P_0,P_1,h) = P_0(h(Y)=1)$ and $\beta_{\I\I}(P_0,P_1,h) = P_1(h(Y)=0)$ the type I and type II errors, respectively. Among such decision functions, the ROC curve 
(also called the tradeoff function~\cite{GaussianDP}) describes the best type II error that can be achieved for a given level of type I error:
\begin{Definition}[ROC] \label{def:ROC}
    For a binary hypothesis testing experiment with distributions $P_0$ and $P_1$, the ROC curve, $f: [0,1] \to [0,1]$, is defined as 
    \begin{align*}
        f(P_0,P_1) (t) = 
        \inf_{h} \left\lbrace  \beta_\II (P_0,P_1,h): \beta_\I (P_0,P_1,h) \leq t \right\rbrace.
    \end{align*}
\end{Definition}
DP can then be described in terms of the ROC curves:
\begin{Theorem}[{~\cite[Theorem 2.1]{TheCompositionTheoremForDP}}] \label{thm:DPasHP}
Given $\eps \geq 0$ and $\delta \in [0,1]$, a mechanism $M$ is $(\eps,\delta)$-DP if and only if for all neighboring databases $D_0$ and $D_1$ and all $t \in [0,1],$
\begin{align*}
    t + e^\eps f(P_0,P_1)(t) \geq 1-\delta, \\
    e^\eps t + f(P_0,P_1)(t) \geq 1 - \delta,
\end{align*}
where $M(D_0) \sim P_0$ and $M(D_1) \sim P_1$.
\end{Theorem} 
The region $(\beta_\I,\beta_\II)$ corresponding to the above constraints will also be referred to as the ``privacy region'' of the mechanism.
The ROC curve corresponding to $(\eps,\delta)$-DP is shown in Figure~\ref{fig:roc_dp} (region in gray).
\begin{figure}[htp]
\centering
\includegraphics[scale=0.3]{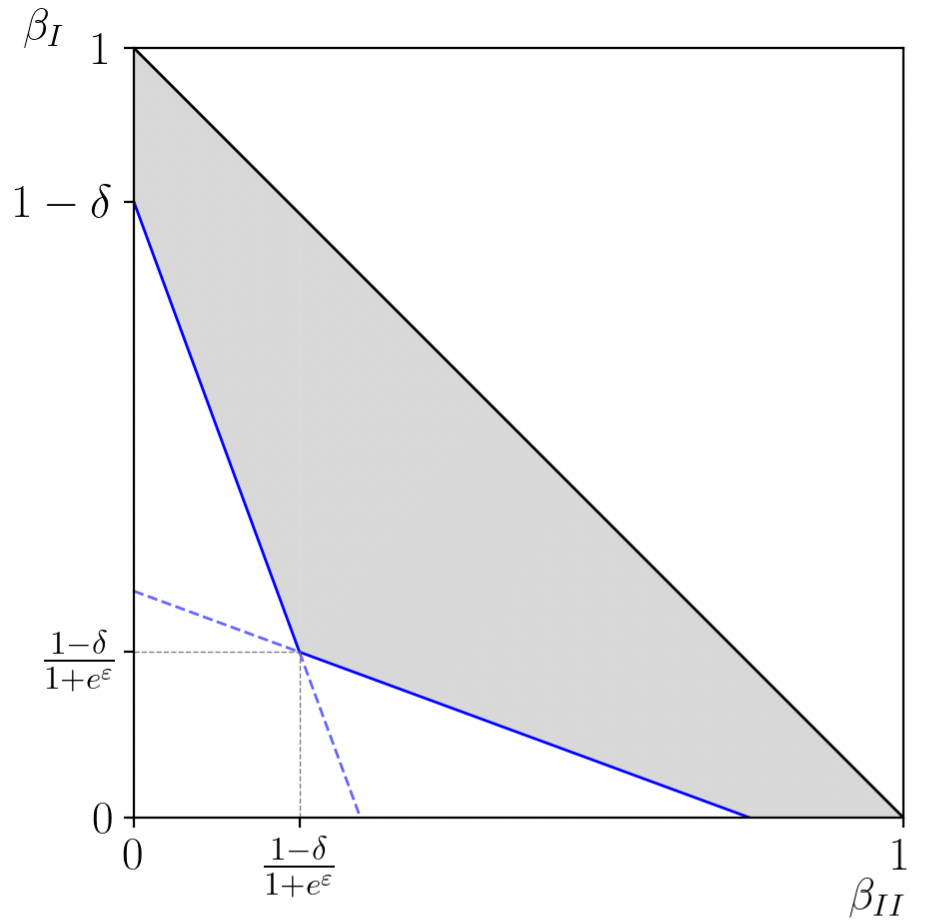}
\caption{ROC corresponding to $(\varepsilon,\delta)$-DP.}
\label{fig:roc_dp}
\end{figure}
We augment the $(\eps,\delta)$-DP framework using the notion of ``total variation of the mechanism'' (which first appeared in \cite{barber2014privacy}). First, recall
\begin{Definition}[Total Variation]
    Given two distributions $P$ and $Q$ over a common alphabet $\Y$, the total variation is defined as follows:
    \begin{align*}
        d_{TV}(P,Q) = \sup_{A \subseteq \Y} (P(A)-Q(A) ).
    \end{align*}
\end{Definition}
\begin{Definition}[Total Variation of a Mechanism]\label{def:TV-mechanism} Given $\eta \in [0,1]$, a mechanism $M$ has total variation less than $\eta$ (or $d_{TV}(M) \leq \eta$ or $\eta$-TV, for short) if, for all neighboring databases $D_0$ and $D_1$,
    \begin{align} \label{eq:def-TV-mechanism}
        d_{TV} \left( P_0, P_1 \right) \leq \eta,
    \end{align}
    where $M(D_0) \sim P_0$ and $M(D_1) \sim P_1$.
\end{Definition}
Total variation is closely related to hypothesis testing: given two distributions $P_0$ and $P_1$, then $d_{TV}(P_0,P_1) \leq \eta$ if and only if\footnote{
This can be seen by noting the following: consider the line with slope -1 tangent to the curve $f$, then its $y$-intercept is given by
\begin{align*}
    -f^\star(-1) & = - \sup_{s \in [0,1]} \{ -s -f(s) \} = \inf_{s \in [0,1]} \{ s + f(s) \} \\
    & = \inf_{A \subseteq \Y} P_0(A) + P_1(A^c) = 1 + \inf_{A \subseteq \Y} P_0(A) - P_1(A) \\
    & = 1 - d_{TV}(P_0,P_1).
\end{align*}
} $f(P_0,P_1)(t) \geq 1-\eta -t$ for all $t \in [0,1]$. As such, $M$ is $\eta$-TV can be rewritten as $(0,\eta)$-DP.
\begin{Corollary} \label{corr:DP-TV-HP}
Given $\eps \geq 0$, $\delta \in [0,1]$, and $\eta \in [0, \delta+ (1-\delta)\frac{e^\eps-1}{e^\eps+1}]$, a mechanism $M$ is $(\eps,\delta)$-DP and $\eta$-TV if and only if, for all neighboring databases $D_0$ and $D_1$ and  all $t \in [0,1]$, As such, $M$ is $\eta$-TV can be rewritten as $(0,\eta)$-DP.
\begin{align}
    t + e^\eps f(P_0,P_1)(t) \geq 1-\delta, \\
    e^\eps t + f(P_0,P_1)(t) \geq 1 - \delta, \\
    t + f(P_0,P_1)(t) \geq 1-\eta, \label{eq:tv_inequality}
\end{align}
where $M(D_0) \sim P_0$ and $M(D_1) \sim P_1$.
\end{Corollary}
\begin{Remark}
The maximum possible value for $\eta$ can be seen from the region in Figure~\ref{fig:roc_dp}, where $\beta_\I +\beta_\II \geq \frac{2(1-\delta)}{1+e^\eps}= 1-\left(\delta+(1-\delta) \frac{e^\eps -1}{e^\eps+1}\right)$.
\end{Remark}
The corresponding ROC curve is shown in Figure~\ref{fig:roc_dp_tv}, with the privacy region shaded in gray. Compared with Figure~\ref{fig:roc_dp}, the effect of introducing the $\eta$-TV constraint corresponds to ``slicing'' the original region (``moving it away'' from the origin). Jia \emph{et al.}~\cite[Theorem 4]{TVD_MIA} concurrently 
studied $\eta$-TV privacy in the hypothesis testing framework. However, their formulation does not simultaneously incorporate both $\eps$-DP and $\eta$-TV (as opposed to this work). Nevertheless, enforcing both $\eps$-DP and $\eta$-TV constraints can be captured within the $f$-DP framework~\cite{GaussianDP}, by choosing $f$ that satisfies the inequalities of Corollary~\ref{corr:DP-TV-HP} with equality.
\begin{figure}[htp]
\centering
\includegraphics[scale=0.3]{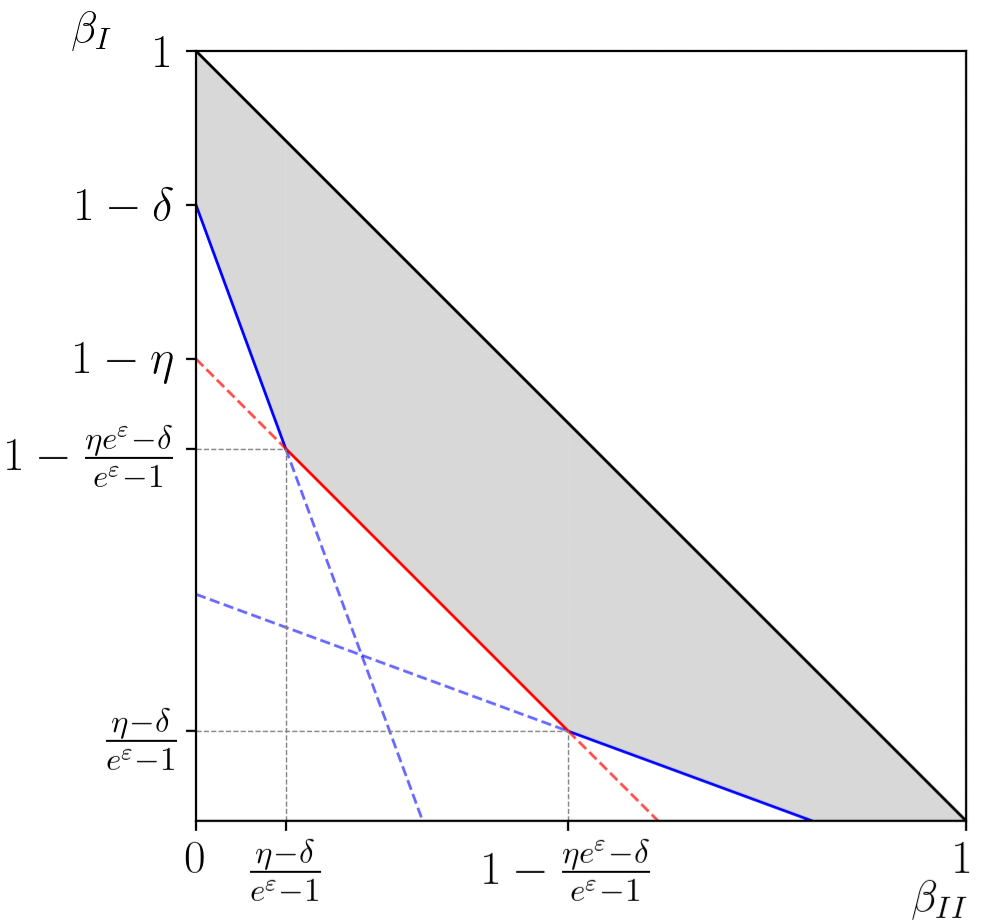}
\caption{ROC corresponding to $(\varepsilon,\delta)$-DP $\eta$-TV.}
\label{fig:roc_dp_tv}
\end{figure}

\section{Main Result: Adaptive Composition} \label{sec:main}
Our main result characterizes exactly the $k$-fold adaptive composition of $(\eps,\delta)$-DP and $\eta$-TV mechanisms. More precisely, the $k$-fold binary hypothesis experiment is defined as follows:
\begin{itemize}
    \item An adversary $\mathcal{A}$ chooses two neighboring databases $D_0$ and $D_1$.
    \item A parameter $b \in \{0,1\}$ is fixed (unknown to $\mathcal{A}$).
    \item $\mathcal{A}$ issues $k$ queries $q_1$, $q_2, \ldots, q_k$, and receives $M_1(D_b)$, $M_2(D_b), \ldots, M_k(D_b)$. 
    The choice of each query (or equivalently, the mechanism $M_i$) may depend on the outputs of previous queries. In particular, letting $\Y_i$ be the output alphabet of $M_i$, then $M_i$ is a mechanism $M_i: \mathcal{D} \times \Y_1 \times \ldots \times \Y_{i-1} \rightarrow \Y_i$, satisfying: for all $y^{i-1} \in Y_1 \times Y_2 \ldots Y_{i-1}$, $M_i(.|y^{i-1})$ is an $(\eps,\delta)$-DP and $\eta$-TV mechanism.
    \item $\mathcal{A}$ guesses $\hat{b} \in \{0,1\}$.
\end{itemize}
Given that each $M_i$ is an $(\eps,\delta)$-DP and $\eta$-TV mechanism, what is the privacy region induced by the above experiment? Allowing the choice of each query to depend on the outcomes of previous queries is referred to as \emph{adaptive composition}.

\begin{Theorem} \label{thm:CompositionTVApproximate} 
For any $\varepsilon \geq 0$, $\delta \in [0,1]$, and $\eta \in [\delta, \delta + \frac{(e^{\varepsilon}-1)(1-\delta)}{1+e^{\varepsilon}}]$, the class of $(\varepsilon,\delta)$-differentially private and $\eta$-total variation mechanisms satisfies
\begin{align*}
    &(j\varepsilon, 1 - (1-\delta)^{k}(1-\delta_j)) \text{-differential privacy with } \\
    &d_{TV} = 1 - (1-\delta)^{k}(1-\delta_0)
\end{align*}
under $k$-fold adaptive composition, for all $j \in \{0, 1, ..., k\}$, where
\begin{align}
     \delta_j   = \!\! \sum_{a = 0}^{k-j-1} \!\! {k \choose a}  \!\! \sum_{\ell=0}^{\lceil\frac{k-j-a}{2}\rceil - 1}  & {k-a \choose \ell} \left(\frac{1-\alpha}{1+e^\varepsilon} \right)^{k-a}  \cdot  \notag \\
     & \alpha^a  \left(e^{(k-\ell-a)\varepsilon} -e^{(\ell+j)\varepsilon} \right) \label{eq:thm-compositionTV}
\end{align}
and $\displaystyle \alpha = 1 - \frac{(\eta-\delta)(1+e^{\varepsilon})}{(1-\delta)(e^{\varepsilon}-1)}$.
\end{Theorem}

The proof follows the machinery developed by Kairouz \emph{et al.}~\cite{TheCompositionTheoremForDP}. In particular, we introduce a ``dominating'' mechanism, i.e., a mechanism which exactly achieves the region described by the equations of Corollary~\ref{corr:DP-TV-HP}. The achievability is then shown by analyzing the (non-adaptive) composition of the dominating mechanism (which admits a simple form). The key component of the converse, similarly to~\cite{TheCompositionTheoremForDP}, is a result by Blackwell~\cite[Corollary of Theorem 10]{BlackwellEquivalentExperiments} which states the following: for two binary hypothesis testing experiments $A$ and $B$, if $f_B(t) \geq f_A(t)$ for all $t \in [0,1]$, then $B$ can be simulated from $A$. This is why the dominating mechanism yields the worst-case degradation (other mechanisms can be simulated from it).
 
\subsection{Comparison}

The composition bound proved by Kairouz \emph{et al.}~\cite{TheCompositionTheoremForDP} states that for any $\varepsilon$ and $\delta \in [0,1]$, the class of $(\varepsilon,\delta)$-differentially private mechanisms satisfies
\begin{align*}
    (j\varepsilon, 1 - (1-\delta)^{k}(1-\delta_j)) \text{-differential privacy}
\end{align*}
under $k$-fold adaptive composition, for all $j \in \{k- 2i: i=0,1,\ldots,\lfloor k/2 \rfloor \}$, where
\begin{align} \label{eq:composition-kairouz}
\delta_{j} = \frac{\sum_{\ell = 0}^{\frac{k-j}{2}-1} {k \choose \ell} \left( e^{(k-\ell)\varepsilon} - e^{(\ell + j)\varepsilon} \right)}{(1+e^{\varepsilon})^{k}}.
\end{align} 
Although the above bound is tight, commonly used mechanisms, like the Laplace mechanism, do not achieve the entire region described in Theorem~\ref{thm:DPasHP}. Taking into consideration the mechanism's total variation leads to a better bound on its privacy region, as shown in  Figure~\ref{fig:laplace_tradeoff} (details of the computations deferred to Section~\ref{subsec:laplace}).

We illustrate the improved composition bound in Figure~\ref{fig:ExampleComposition}, with $\eps = 1$, $\delta=0$, and $ \eta= 0.7\frac{e^{\eps}-1}{1+e^{\eps}}$ (i.e., $\alpha = 0.3$). Note that the refined composition bound involves lines with slopes of the form $-e^{\pm j\eps}$ for $j \in \{0,1,...,k\}$, while the bound introduced in~\cite{TheCompositionTheoremForDP} only involves values of $j$ such that $0\leq j \leq k$ and $j$ and $k$ have the same parity. Even if one were to only observe the lines that correspond to values of $j$ that have the same parity as $ k $ (blue vs green in Figure~\eqref{fig:ExampleCompositionDominatingMechanisms}), the refined bound still improves on the previous bound.
\begin{figure}[htp]
  \centering
  \includegraphics[scale=0.3]{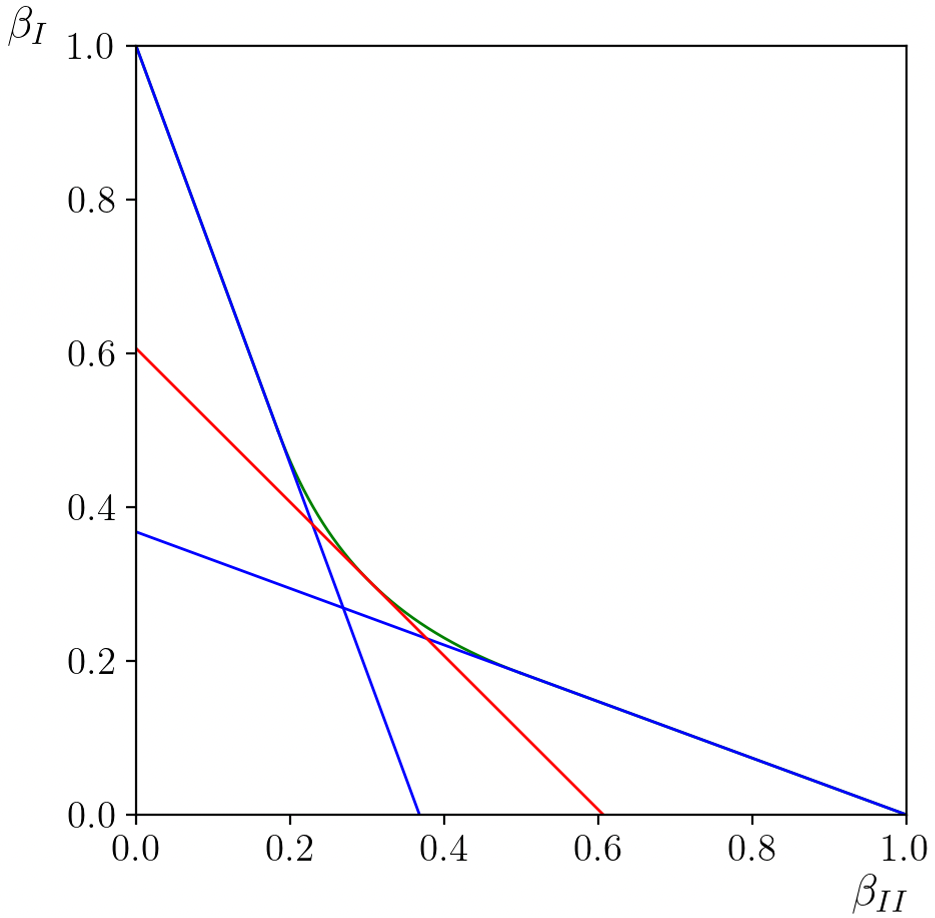} 
  \caption{The Laplace mechanism's exact tradeoff function is shown in green (for $\eps=1$). The blue lines correspond to the differential privacy constraints. The red line corresponds to $\beta_{\I} = 1 - \eta_{\mathrm{Lap}} - \beta_{\II}$ where $\eta_{\mathrm{Lap}} = 1-e^{-\eps/2}$ is the total variation of the Laplace mechanism.}
  \label{fig:laplace_tradeoff}
\end{figure}

\begin{figure}[htp]
\centering
\subfloat[ $k=1.$
]{
\label{fig:ExampleDominatingMechanisms}
\includegraphics[scale=0.3]{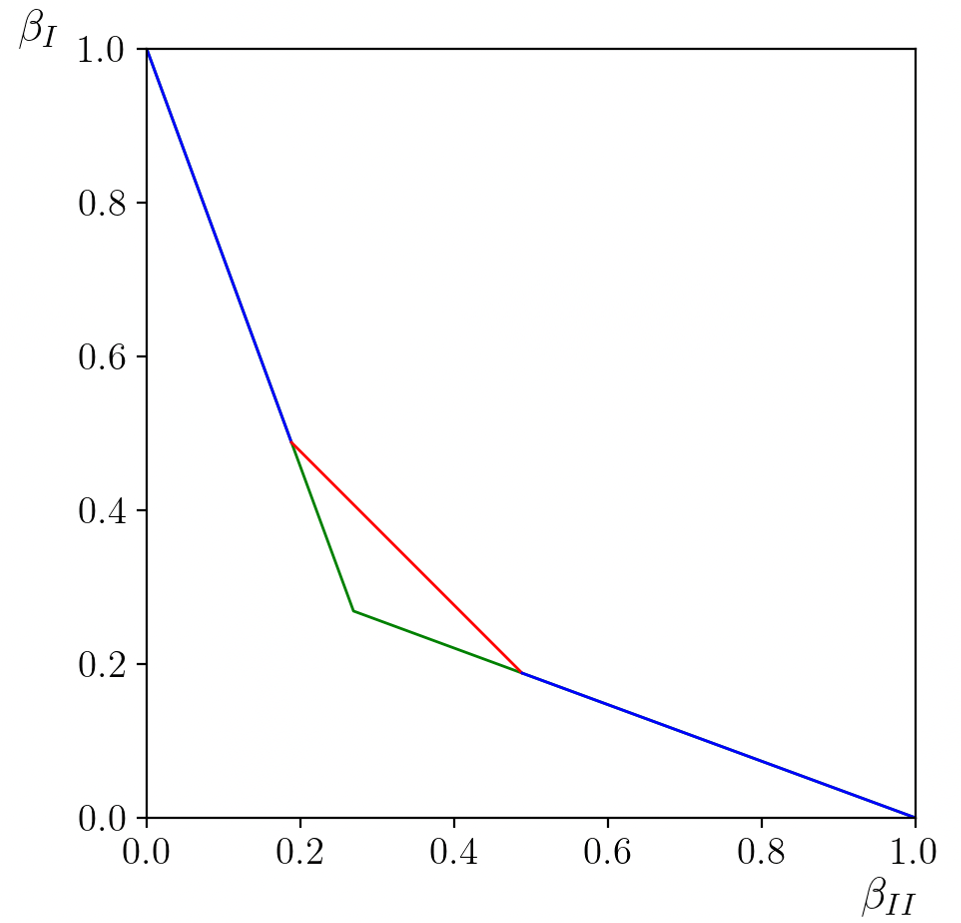}
}  \\
\subfloat[$k=5.$]{
\label{fig:ExampleCompositionDominatingMechanisms}
\includegraphics[scale=0.3]{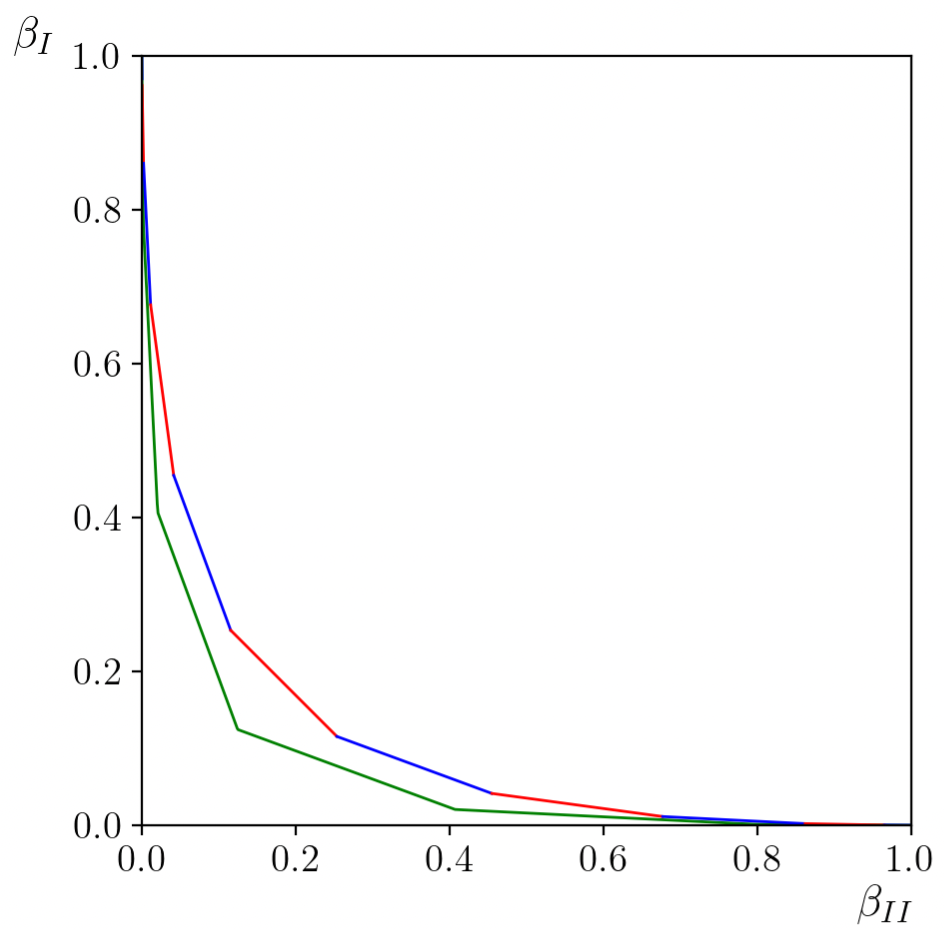}
}
\caption{The induced tradeoff functions of the dominating mechanisms for $\eps = 1$, $\delta = 0$, and $\alpha=0.3$. The green line shows the  tradeoff function of $\eps$-DP (given by equation~\eqref{eq:composition-kairouz}). The blue and red lines plot the tradeoff given by Theorem~\ref{thm:CompositionTVApproximate}, where they correspond, respectively, to odd (same parity as $k=5$) and even choices of $j$ in equation~\eqref{eq:thm-compositionTV}.}
\label{fig:ExampleComposition}
\end{figure}

\begin{Remark}
    Jia \emph{et al.}~\cite[Theorem 12]{TVD_MIA} derived the $k$-fold composition for any $\eta$ with $\eps = \log \frac{1 + \eta}{1 - \eta}$. However, this reduces to the standard $\eps$-DP setting ($\eta = \frac{e^\eps-1}{e^\eps+1}$), which has already been derived by Kairouz \emph{et al}.~\cite{TheCompositionTheoremForDP}. The same applies to~\cite[Theorem 18]{TVD_MIA}. Hence, contrary to our result, the $\eta$ parameter is not used to further bound the privacy region and obtain tighter composition results.
\end{Remark}

\subsection{Discussion}

\subsubsection{Membership Inference Attacks}

Prior works have utilized the notion of total variation to capture an adversary's advantage and to bound vulnerability to Membership Inference Attacks (MIAs), in which the attacker tries to determine whether a given record was part of the model's training data.
 Chatzikokolakis \emph{et al.} \cite{BayesSecurity} studied the notion of Bayes security (denoted $\beta^*$), a metric inspired by cryptographic advantage, and which is proven to be the complement of the total variation of the two maximally distant rows of the channel. The authors also proved a parallel composition result, which states that the composition of two channels that are $\beta^*_1$-secure and $\beta^*_2$-secure yields a $\beta^*_1 \beta^*_2$-secure channel.
    We can recover this result for the special case  $\beta^*_1 =  \beta^*_2$ by considering the total variation of the composed dominating $\eta$-TV mechanism with $\eta = 1-\beta^*$ (or, equivalently, the dominating $(0,\eta)$-DP mechanism), which is $ 1 - (1 - \eta)^k$ (this can be seen by noting that the total variation under composition is $d_{TV} = 1 - (1 - \delta)^k (1 - \delta_0)$, and that $\delta_0 = 0$ in equation~\eqref{eq:composition-kairouz} for $\eps=0$ and $\delta=\eta$). 
 Kulynych \emph{et al.} \cite{DisparateVulnerability} studied the phenomenon of disparate vulnerability to MIAs, i.e., the unequal success rate of MIAs against different population subgroups (or the difference between subgroup vulnerability to MIAs). The authors showed that the worst-case vulnerability of a subgroup to MIAs is bounded by the total variation of the mechanism, and consequently, so is the disparity.

Considering the importance of MIA attacks and subsequently the importance of the total variation of the mechanism, it is crucial to keep track of the evolution of the $\eta$ parameter under composition. Indeed, our result enables accurate tracking of $\eta$ by coupling the $\eta$-TV and $\eps$-DP constraints. By contrast, studying the worst-case evolution of $\eta$ without taking $\eps$ into consideration would yield much looser bounds: $\eta$ would evolve as $1-(1-\eta)^k$. Moreover, from a mechanism design point of view, our framework enables greater flexibility. In particular, when comparing two mechanisms, a designer may want to consider both their $(\eps,\delta)$-parameters as well as their $\eta$-parameters to have a more complete comparison.

 \subsubsection{Differentially Private Stochastic Gradient Descent}

The more refined analysis of the composition of privacy-preserving mechanisms is crucial in studying algorithms like DP-SGD, which involves subsampling followed by a large number of compositions (so that standard composition results yield loose bounds). We derive a subsampling result in Section~\ref{sec:subsampling}, and couple it with our composition result to analyze DP-SGD in Section~\ref{sec:SGD}.
\subsection{Proof} \label{refined_dominating_mechanism}
We propose an $(\eps,\delta)$-differentially private mechanism with total variation $\eta$ that achieves the entire privacy region corresponding to $(\eps,\delta)$-DP and $\eta$-TV (described in Corollary~\ref{corr:DP-TV-HP}, shown in Figure~\ref{fig:roc_dp_tv}), and that therefore \emph{dominates} all $(\eps,\delta)$-differentially private mechanisms with total variation $\eta$. We then derive its privacy region under non-adaptive $k$-fold composition, which consists of the intersection of several $(\eps_j,\delta_j)$-DP privacy regions. Finally, we show that this derived privacy region is the largest achievable privacy region under $k$-fold (adaptive) composition of $(\eps,\delta)$-DP and $\eta$-TV mechanisms, hence proving the tightness of our result.
\subsubsection{Dominating Mechanism} 
Kairouz~\cite{TheCompositionTheoremForDP} \emph{et al.} introduced a mechanism, which we will denote by $M_{\eps,\delta}: \X^m \to \{0,1\}$, that dominates all $(\varepsilon,\delta)$-differentially private mechanisms in the following sense: for any $(\eps,\delta)$-DP mechanism $M: \X^m \to \Y$, for any pair of databases $D_0$ and $D_1$,  
 there exists a (possibly randomized) mapping $T: \{0,1\} \to \Y$ such that, $M(D)$ and $T(M_{\eps,\delta}(D))$ have the same distribution, where $D \in \{D_0,D_1\}$. 
We adopt a similar approach, and provide a mechanism that dominates all $(\varepsilon,\delta)$-differentially private mechanisms with total variation $\eta$. Due to lack of space, we provide herein the proof of the case $\delta = 0$ only and include the details on how we generalize the result to approximate differential privacy in Appendix~\ref{app:general-delta}. 

The following mechanism (as in~\cite{TheCompositionTheoremForDP}) does not depend on the database entries or the query; it only depends on the hypothesis, and will be denoted by $M_{\eps,0,\eta}$. It outputs $X_0 \sim P_0$ in the case of the null hypothesis, and $X_1 \sim P_1$ in the case of the alternative hypothesis, where $P_0$ and $P_1$ are defined as follows, for $\varepsilon \geq 0$ and $\alpha \in [0,1]$:

\begin{align} \label{eq:def-P0-P1}
    \! P_0 \! = \! \begin{cases}
    \frac{(1-\alpha)e^{\varepsilon}}{(1+e^{\varepsilon})}, &  x=0,\\
    \alpha, &  x=1,\\
    \frac{(1-\alpha)}{(1+e^{\varepsilon})}, &  x=2,
    \end{cases} 
~
    P_1 \!= \! \begin{cases}
    \frac{(1-\alpha)}{(1+e^{\varepsilon})}, & x=0,\\
    \alpha, & x=1,\\
    \frac{(1-\alpha)e^{\varepsilon}}{(1+e^{\varepsilon})}, &  x=2.
    \end{cases} 
\end{align}
The above mechanism's total variation is $d_{TV}(P_0,P_1) = \frac{(e^{\varepsilon}-1)(1-\alpha)}{(1+e^{\varepsilon})}$. So for a target $\eta \in [0, \frac{e^\eps-1}{e^\eps+1}]$, we set $\alpha  = 1- \eta \frac{e^\eps+1}{e^\eps-1}$. For $\eta =  \frac{e^\eps-1}{e^\eps+1}$, we recover $M_{\eps,0}$.
It is easy to verify that $f(P_0,P_1)$ exactly achieves the region described in Corollary~\ref{corr:DP-TV-HP}.

\subsubsection{Composed Dominating Mechanism}
\label{composed_dominating_mechanism_epsilon_0}
We show that the region described in Theorem~\ref{thm:CompositionTVApproximate} is achievable by analyzing the non-adaptive composition of $M_{\eps,0,\eta}$, i.e., in the $k$-fold binary hypothesis experiment, we set $M_i = M_{\eps,0,\delta}$ for all $i=1,2,\ldots,k$.

Since the privacy region of the composed mechanism is a convex set, we can describe its boundaries with a set of lines tangent to its envelope. For a slope $-\gamma$, we want to find the largest shift $1-\delta'$ such that the line $\beta_{\I}= -\gamma \beta_{\II} + 1-\delta'$ is below the tradeoff curve. 

For some acceptance set $A$, consider the point $(\beta_{\II}(A),\beta_{\I}(A))$ included in the composed mechanism's privacy region. 
We want $-\gamma \beta_{\II}(A) + 1 - \delta' \leq \beta_{\I}(A)$:
\begin{align}
    \delta' & \geq 1 - \beta_{\I}(A) - \gamma \beta_{\II}(A) \label{inequality_delta}\\
    & = P_{0}(A) - \gamma P_{1}(A).
\end{align}
We choose the smallest $\delta'$ such that the inequality~\eqref{inequality_delta} is satisfied for all $A$, i.e.,
\begin{align*}
\delta' = \max_{\substack{A \in \mathcal{X}}} \{P_0(A) - \gamma P_1(A)\}.
\end{align*}
We represent the output of a $k$-fold composition experiment as a sequence $s^{k} \in \{0,1,2\}^k$, and use the following notations:
\begin{align}
    a = | \{i: s_i^k = 1\}|, \text{ and }
    \ell = | \{i: s_i^k = 2\}|. \label{output_sequence} 
\end{align}
Let $\tilde{P_0}$ and $\tilde{P_1}$ be the probabilities of obtaining a sequence $s^k$ under $k$-fold composition when the true database is $D_0$ and $D_1$, respectively. Since the outputs of the mechanisms are i.i.d.,
\begin{align*}
    \tilde{P_0}(s^k) & = P_0(X=0)^{k-\ell-a} P_0(X=1)^{a} P_0(X=2)^{\ell}  \\
& = \left (\frac{1-\alpha}{1+e^\varepsilon} \right)^{k-a} e^{(k-\ell-a)\varepsilon} \alpha^a,  
\end{align*}
and
\begin{align*}
 \tilde{P_1}(s^k) & = P_1(X=0)^{k-\ell-a} P_1(X=1)^{a} P_1(X=2)^{\ell} \\
&= {\left( \frac{1-\alpha}{1+e^\varepsilon} \right)}^{k-a} e^{\ell \varepsilon} \alpha^a.
\end{align*}
By the Neyman-Pearson lemma~\cite[Theorem 11.7.1]{thomascover}, points on the ROC curve are achieved by tests of the form $\{s^k: \frac{\tilde{P}_0(s^k)}{\tilde{P}_1(s^k)} > T \}$, for some threshold $T$. Letting
\begin{align*}
    E \equiv \{0 \leq \tilde{\varepsilon} \leq \infty: \tilde{P_0}(s_k) = e^{\tilde{\varepsilon}} \tilde{P_1}(s_k) \text{ for some } x \in \mathcal{X} \},
\end{align*}
any $\tilde{\varepsilon} \notin E$ does not contribute to the curve.
\begin{Remark}
    It is sufficient to consider $\tilde{\eps} \geq 0$ as the tradeoff function is symmetric due to the symmetry of $P_0$ and $P_1$. The symmetry is then accounted for by the symmetry of $(\eps,\delta)$-DP.
\end{Remark}
Noting that 
$$
\frac{\tilde{P_0}(s_k)}{\tilde{P_1}(s_k)} = \frac{e^{(k-\ell-a)\varepsilon}}{e^{\ell \varepsilon}}, \text{ it follows } \tilde{\varepsilon} = (k-2\ell-a)\varepsilon.$$
Therefore, the distinct values that $\tilde{\varepsilon}$ can take are of the form $\varepsilon_j = j \varepsilon$ where $j \in \{0,1,...,k\}$. 
To every $\varepsilon_j$ corresponds a $\delta_j$ of the form
\begin{align*}
    \delta_j = \max_{\substack{A \in \mathcal{X}^k}} \{\tilde{P_0}(A) - e^{\varepsilon_j} \tilde{P_1}(A)\},
\end{align*}
where $A$ is the acceptance region.
For a fixed $j$, we consider the sequences $s_k$ such that $\tilde{P_0}(s_k) - e^{j\eps} \tilde{P_1}(s_k) > 0$, i.e. \\
\begin{align*}
    {\left (\frac{1-\alpha}{1+e^\varepsilon} \right)}^{k-a} \alpha^a \left( e^{(k-\ell-a)\varepsilon}  -  e^{(\ell+j)\varepsilon} \right) & > 0.
\end{align*}
There are ${k \choose a} {k-a \choose \ell}$ sequences of the form described in (\ref{output_sequence}), but we are only interested in the sequences that satisfy $ e^{(k-\ell-a)\varepsilon}  -  e^{(\ell+j)\varepsilon} > 0$, or $k-j>2\ell+a$, hence the summation bounds in equation~\eqref{eq:thm-compositionTV}. 

\subsubsection{Converse} Since the tradeoff function of $M_{\eps,0,\eta}$, $f_{{\eps,0,\eta}}$, matches with equality the constraints given in Corollary~\ref{corr:DP-TV-HP}, the tradeoff function of any $\eps$-DP and $\eta$-TV mechanism $M$ must satisfy $f_M(t) \geq f_{{\eps,0,\eta}} (t)$ for all $t \in [0,1]$. Hence by Blackwell's result~\cite[Corollary of Theorem 10]{BlackwellEquivalentExperiments}, $M$ can be simulated from $M_{\eps,0,\delta}$. Given this fact, one can follow exactly the converse proof in~\cite{TheCompositionTheoremForDP} to show that the output of the $k$-fold composition of any $\eps$-DP and $\eta$-TV mechanisms can also be simulated using the (non-adaptive) $k$-fold composition of $M_{\eps,0,\eta}$, hence the corresponding tradeoff function will be greater than the composition of $M_{\eps,0,\eta}$. The intuition is the following: Let $M_{D_1}$ and $M_{D_2}$ be two instances of the dominating $(\eps,\delta)$-DP with $\eta$-TV mechanism.
The adversary chooses a mechanism $M_1$, which is $(\eps,\delta)$-DP with $\eta$-TV, observes its output, and based on it, chooses $M_2$, which is also $(\eps,\delta)$-DP with $\eta$-TV.
$(M_1,M_2)$ can be simulated from $(M_{D_1},M_{D_2})$, since $M_1$ can be simulated from $M_{D_1}$, and $M_2$ can be simulated from $M_{D_2}$.


\section{Subsampling} \label{sec:subsampling}

Subsampling is a simple method to ``amplify'' privacy guarantees: before answering a query on a given database of size $n$, first choose uniformly at random a subset of size $m$, $1 \leq m \leq n$. This procedure will be denoted by $\Sample_m$.  Then, compute the query answer on the sampled database.

\begin{Proposition} \label{prop:subsampling}
    Given $\eps >0$, $\delta \in [0,1]$, and $\eta \in [0,1]$, if $M: \X^m \to \Y$ is ($\eps$,$\delta$)-DP and $\eta$-TV, then the subsampled mechanism $M \circ \Sample_{m}: \X^n \to \Y$ is $p\eta$-TV and ($\log \left( 1+ p \left(e^\eps -1\right) \right), p \delta$)-DP  on $\X^n$, where $m \leq n$ and $p = \frac{m}{n}$. Moreover, this result is tight.
\end{Proposition}
The proof (deferred to Appendix~\ref{app:subsampling}) is based on the privacy amplification by subsampling result for ($\eps$,$\delta$)-DP appearing in \cite{KLNRS, Smith_2009, Chaudhuri_Mishra,balle:2018} and stated in ~\cite[Theorem 29]{steinke2022composition}. As such, $(\eps,\delta)$-DP with $\eta$-TV is closed under subsampling.
Moreover, as can be seen from the proof, it holds in general that if $M$ is $\eta$-TV, then $M \circ \Sample_m$ is $p \eta$-TV. This result was concurrently derived by Jia \emph{et al.}~\cite{TVD_MIA}, but we include a simple proof herein and further show tightness.

\section{Asymptotic Behaviour} \label{sec:limit}

Inspired by the hypothesis testing view of differential privacy~\cite{TheCompositionTheoremForDP,DiffPrivacyAsHT}, Dong \emph{et al.}~\cite{GaussianDP} introduced a new privacy measure, called $f$-differential privacy, based directly on tradeoff functions (i.e., ROC curves):
\begin{Definition}
    [$f$-DP~\cite{GaussianDP}]
    Let $f: [0,1] \to [0,1]$ be a convex, continuous, non-increasing function, satisfying $f(x) \leq 1-x$. A mechanism $M$ is $f$-DP if for all neighboring databases $D_0$ and $D_1$ and $t \in [0,1]$,
    \begin{align}
        \label{eq:def-fDP}
        f(P_0,P_1)(t) \geq f(t),
    \end{align}
    where $M(D_0) \sim P_1$ and $M(D_1) \sim P_1$.
\end{Definition}

As such, $(\eps,\delta)$-DP, as well as $(\eps,\delta)$-DP with $\eta$-TV, are special cases of $f$-DP, with $f$ given by the corresponding curves shown in Figure~\ref{fig:roc_dp_tv}. 

Notably, the composition of $f$-DP mechanisms is also $f$-DP (with a potentially different $f$), i.e., $f$-DP is closed under composition. Moreover, the authors~\cite{GaussianDP} show that composition for $f$-DP follows a central limit theorem-like behavior in the following sense. Given $\mu \in \mathbb{R}$, let
\begin{align}
    \label{eq:def-Gaussain-ROC}
    G_\mu (\alpha) = f( \mathcal{N}(0,1), \mathcal{N} (\mu,1))(\alpha),
\end{align}
for $\alpha \in [0,1]$.
Then, the composition of $n$ $f$-DP mechanisms can be approximated by $G_\mu$ (with an appropriate $\mu$)~\cite[Theorem 3.4]{GaussianDP}. Indeed, under suitable assumptions, the limit of the composition is equal to $G_\mu$~\cite[Theorem 3.5]{GaussianDP}.  For instance, letting $\eps_n = \frac{\mu}{\sqrt{n}}$, then the composition of $n$ $\eps_n$-DP mechanisms converges to $G_\mu$~\cite[Theorem 3.6]{GaussianDP}. Here, we provide an analogous result for $\eps$-DP with $\eta$-TV (proof deferred to Appendix~\ref{app:central-limit}): 
\begin{Theorem}\label{thm:central-limit}
    Consider a triangular array $\{(\eps_{ni},\eta_{ni}): 1 \leq i \leq n \}_{n=1}^{\infty}$ such that $\displaystyle  \lim_{n \rightarrow \infty} \sum_{i=1}^n \eps_{ni} \eta_{ni} = \frac{\mu^2}{2}$ for some $\mu \geq 0$, and $ \displaystyle \lim_{n \rightarrow \infty} \max_{1 \leq i \leq n} \eps_{ni} =0$. Let $M_{ni}: \X^m \to \Y_{ni}$ be an $\eps_{ni}$-DP and $\eta_{ni}$-TV mechanism, and let $M_n (D) = (M_{n1}(D), \ldots, M_{nn}(D))$, $D \in \X^m$. Then, there exists a sequence of functions $\{f_n\}$ such that $M_n(D)$ is $f_n$-DP and
    \begin{align} \label{eq:thm-central-limit}
        \lim_{n \rightarrow \infty} f_n = G_\mu,
    \end{align} 
    uniformly over $[0,1]$, where $G_\mu$ is defined in Eq.~\eqref{eq:def-Gaussain-ROC}.
\end{Theorem}

For pure differential privacy, $\eta = \frac{e^\eps-1}{e^\eps+1} \approx \frac{\eps}{2}$ for small $\eps$, so that the above result matches Theorem~3.6 of~\cite{GaussianDP}. However, if for small $\eps$, $\eta \approx \frac{\eps}{c_0}$ for some $c_0 > 2$, then: if Theorem~3.6 of~\cite{GaussianDP} yields a parameter $\mu$, Theorem 3 will yield $\mu \sqrt{ \frac{2}{c_0}}$. It is worth noting that one can always reduce the total variation of a mechanism by a factor $(1-\delta')$ by passing its output through an erasure channel with parameter $\delta'$.
Moreover, let $c_1 \in (0,1/2)$, $\eps_{ni} = \frac{\mu}{n^{c_{1}}}$, and $\eta_{ni} = \frac{\mu}{2 n^{1-c_1}}$, for all $1 \leq i \leq n$. Then, the conditions of the theorem are satisfied, and the limit of the composition is $G_\mu$. As opposed to decaying as $1/\sqrt{n}$, $\eps_n$ could decay with an arbitrarily small power of $n$, as long as this is ``compensated'' for with a faster decay for the total variation.
Finally, it is worth noting that, as compared to the $f$-DP framework, $(\eps,\delta)$-DP with $\eta$-TV is significantly simpler as we keep track of only three parameters, while still yielding considerable advantages over $(\eps,\delta)$-DP. Moreover, it is easier to work with from a design perspective (choosing $\eps$ and $\eta$ parameters versus choosing a function).
\begin{Remark}
  Although Theorem~\ref{thm:central-limit} was stated for non-adaptive composition, one could allow adaptive choices of $M_{ni}$ and retain the same result (this is due to a tightness result by Dong \emph{et al.}~\cite[Theorem 3.2]{GaussianDP}).   
\end{Remark}

\section{Applications} \label{sec:app}
In this section, we compute the total variation for the Laplace, Gaussian, and staircase mechanisms, to demonstrate improvement and/or simplicity of analysis for commonly used mechanisms. Using well-known bounds on the total variation (e.g., Pinsker's inequality), one could also derive simple closed-form bounds on the total variation of the composition. While one could efficiently approximate the composed tradeoff functions of mechanisms using Fourier-based methods~\cite{TightDP_KJPH, NumericalCompositionGopi, OptimalAccountingCharacteristicFunction}, doing so requires the CDF of the ``privacy loss'' random variable and may not lead to a closed form. We also discuss an additional use-case for keeping track of the total variation of a mechanism in the context of Membership Inference Attacks.

\subsection{Laplace Mechanism} \label{subsec:laplace}
The probability density function of a Laplace distribution centered at $0$ and with scale $b$ is
\begin{align}
    Lap(x|b) = \frac{1}{2b} \exp(-\frac{|x|}{b}),
\end{align}
and will be denoted by $Lap(b)$. 
For a (query) function $q: \mathcal{X}^{m} \rightarrow \mathbb{R}$, the output of the Laplace mechanism, denoted $M_\mathrm{Lap}$, is $M_\mathrm{Lap}(X)= q(X) + Y$ where $Y \sim Lap(\Delta/\eps)$, and the query sensitivity $\Delta$ is defined as
\begin{equation}
    \Delta = \max_{\substack{D \text{ and } D' \text{ neighbours}}} |q(D)-q(D')|.
\end{equation}
$M_\mathrm{Lap}$ is $\eps$-DP and $d_{TV}(M_\mathrm{Lap}) = 1-e^{-\frac{\eps}{2}}$ (derivation in Appendix \ref{app:laplace}). This result was concurrently derived by Jia \emph{et al.} \cite{TVD_MIA} and Chatzikokolakis \emph{et al.} \cite{BayesSecurity}.
\subsection{Gaussian Mechanism}
The total variation of a Gaussian mechanism operating on a statistic $\theta$ with variance $\sigma^2 = \frac{sens(\theta)^2}{\mu^2}$ for $\mu \geq 0$ is
\begin{align}
    d_{TV} & = 2\Phi\Bigl(\frac{\mu}{2}\Bigr)-1,
\end{align}
where $\Phi$ is the Standard Normal CDF. This result was proven by Jia \emph{et al.} \cite{TVD_MIA} and Chatzikokolakis \emph{et al.} \cite{BayesSecurity}, and appeared implicitly in \cite{Improving-the-Gaussian-Mechanism}. We provide an alternative derivation in Appendix \ref{app:gaussian}.
\subsection{Staircase Mechanism}
Geng \emph{et al.}~\cite{TheOptimalMechanismInDifferentialPrivacy} studied privacy-utility tradeoffs using differential privacy. They demonstrated that, for a large family of utility functions, the optimal noise-adding mechanism is the ``staircase mechanism''. That is, the noise is drawn from a staircase-shaped probability distribution with probability density function
\begin{equation*} f_{\gamma}(x) = 
  \left\{
    \begin{aligned}
        & a(\gamma), && x \in [0,\gamma\Delta), \\
        & e^{-\eps} a(\gamma), &&  x \in [\gamma\Delta,\Delta), \\
        & e^{-k\eps}f_{\gamma}(x-k\Delta), &&  [k\Delta,(k+1)\Delta), ~ k \in \mathbb{N}, \\
        & f_{\gamma}(-x), && x<0.
    \end{aligned}
  \right.
\end{equation*}
where $\displaystyle a(\gamma) = \frac{1 - e^{-\varepsilon}}{2\Delta(\gamma + e^{-\varepsilon}(1 - \gamma))}$, $\gamma \in \mathbb{R}^{+}$. 
The total variation of the staircase mechanism is given by (derivation in Appendix~\ref{app:staircase}):
\begin{equation*} d_{TV}  \!\! = \!\! 
  \left\{
    \begin{aligned}
        & \frac{(1 - e^{-\eps}) (2 \gamma (1-e^{-\varepsilon}) + e^{-\varepsilon})}{2(\gamma + e^{-\varepsilon}(1 - \gamma))}, &\gamma \in [0,\frac{1}{2}), \\
        & \frac{1 - e^{-\varepsilon}}{2(\gamma + e^{-\varepsilon}(1 - \gamma))}, & \gamma \in [\frac{1}{2}, \infty).
    \end{aligned}
  \right.
\end{equation*}

Remarkably, the staircase mechanism with total variation $\eta$ achieves exactly the privacy region corresponding to $(\eps,0)$-DP and $\eta$-TV (proof in Appendix~\ref{app:staircase-region}).

For $\gamma = \frac{1}{2}$, $d_{TV} = \frac{e^{\varepsilon} - 1}{e^{\varepsilon}+1}$, which corresponds to the total variation of the $(\varepsilon,0)$-differentially private dominating mechanism introduced by Kairouz \emph{et al.}~\cite{TheCompositionTheoremForDP}, or to that of the mechanism that we introduced in Section~\ref{refined_dominating_mechanism} for $\alpha = 0$. For $\gamma=0.0139$, $\alpha \approx 0.3$, the corresponding tradeoff function consists of the blue and red lines in Figure~\ref{fig:ExampleDominatingMechanisms}, and therefore the composition exactly corresponds to the blue and red lines in Figure~\ref{fig:ExampleCompositionDominatingMechanisms}. 

\subsection{Differentially Private Stochastic Gradient Descent} \label{sec:SGD}

We consider the noisy SGD algorithm introduced in~\cite{abadi:16}. Given a dataset of size $n$, at every step, noisy SGD generates a uniformly random subsample of size $m$, computes the corresponding gradient, clips it (in $\ell_2$-norm), then adds i.i.d Gaussian noise to each component of the gradient. Dong \emph{et al.}~\cite{GaussianDP} consider and analyze the same setup for comparison purposes, which we adopt here. 

Given specific parameters of the noisy SGD algorithm (batch size $m = 256$, learning rate $=0.25$, and clipping threshold $=1.5$), Dong \emph{et al.}~\cite{GaussianDP} show that every step satisfies $\frac{1}{1.3}$-Gaussian differential privacy (GDP)~\cite[Definition 2.6]{GaussianDP}. For every choice of $\eps$, one obtains a corresponding $\delta$ using~\cite[Corollary 2.13]{GaussianDP}. In particular, if a mechanism is $\mu$-GDP then it is$(\eps, \delta(\eps))$-DP for
\begin{align}
    \delta(\eps) = \Phi\left( -\frac{\eps}{\mu}+ \frac{\mu}{2} \right)- e^\eps \Phi \left( -\frac{-\eps}{\mu} - \frac{\mu}{2} \right),
\end{align}
where $\Phi$ is the CDF of the standard normal distribution~\cite[Corollary 2.13]{GaussianDP}. As such,
\begin{align}
\eta  = \delta(0) = \Phi \left( \frac{\mu}{2}\right) - e^\eps \Phi\left( -\frac{\mu}{2} \right).
\end{align}

For every choice of $(\eps,\delta)$, we may then use Theorem~\ref{thm:CompositionTVApproximate} to derive an outer bound on the composition of noisy SGD iterations. Moreover, we consider the intersection of all such regions (for all chosen pairs $(\eps,\delta)$). The result is illustrated in Figure~\ref{fig:SGD}.
\begin{figure}[htp]
    \centering
    \includegraphics[scale = 0.3]{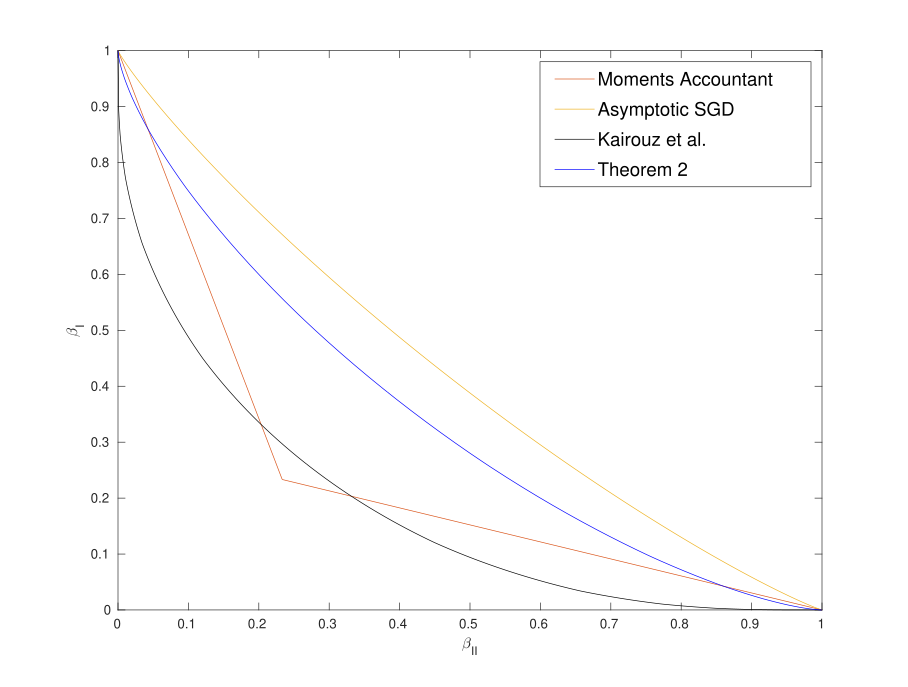}
    \caption{Comparison of the resulting composition bounds: dataset size $n = 60,000$, batch size $m = 256$, learning rate $ = 0.25$, clipping threshold $= 1.5$. Noisy-SGD algorithm running for 15 epochs. Moments accountant refers to the method developed by Abadi \emph{et al.}~\cite{abadi:16} which yields $(1.19,10^{-5})$-DP; Asymptotic SGD corresponds to the asymptotic result of~\cite[Corollary 5.4]{GaussianDP} (both of these curves were retrieved from~\cite{GaussianDP}); Theorem 2 is applied for all $\eps \in [0.5, 3.4]$ in steps of 0.1 while using $\eta$ of the mechanism; Kairouz \emph{et al.}~\cite{TheCompositionTheoremForDP} is applied for the same values of $\eps$.}
    \label{fig:SGD}
\end{figure}

As can be seen from the figure, our method yields a significant improvement over both Abadi \emph{et al.}'s as well as Kairouz \emph{et al.}'s result~\cite{abadi:16,TheCompositionTheoremForDP}. Although our region is looser than the region of Dong \emph{et al.}~\cite{GaussianDP}, their result only holds asymptotically for a specific relationship between the sampling rate and the number of epochs, whereas we obtain a finite-$n$ analysis. It is worth noting that our computations involve an approximation of Theorem~\ref{thm:CompositionTVApproximate} to avoid numerical issues. In particular, we used the method of types~\cite{thomascover}, and approximate the probability of a type using KL divergence~\cite[Exercise I.2]{polyanskiy:23}

\section{Total Variation with Local Differential Privacy}
\label{sec:localDP}

The local privacy setting refers to the case in which there is no ``trusted'' data curator. Hence, every data owner shares with the curator \emph{noisy} versions of their data. In this setting, local differential privacy has emerged as the main framework to guarantee that the noisy versions are indeed private~\cite{duchi2014local}. Herein, we have ``opposing'' goals of the data curator and the data owners. For instance, suppose the data corresponds to ``web browsing history'', then every data owner would like to keep their browsing history private, while the (untrusted) data curator (who is typically offering a service to the data owners) would like to perform statistical analyses on the data. As such, privacy-utility tradeoffs correspond to finding a  mechanism that satisfies a given privacy constraint, yet maximizes a certain statistical utility. One can pose a simplified abstract problem as follows: given two distributions $P_0$ and $P_1$, and a privatization mechanism $Q$, let $M_0$ and $M_1$ be the induced distributions. Suppose the data curator measures utility through a function $U(M_0 || M_1)$. Then we need
\begin{align}
    \max_{Q: Q \text{ satisfies privacy constraint}} U(M_0 || M_1).
\end{align}
A common choice for $U$ is an $f$-divergence (e.g., KL, $\chi^2$, etc.). In practice, $\eps$-LDP analysis often lead to pessimistic results~\cite{duchi2014local,BayesSecurity}. Moreover, the analysis of commonly-used mechanisms (such as $M$-ary randomized response) may be loose in the LDP framework. 

Herein, we study the notion of total variation of a mechanism in the local privacy setting, while maintaining the $\eps$-LDP constraint. This offers more precise results for analyzing mechanisms, and enables greater flexibility in mechanism design (for instance, one may increase $\eps$ and decrease $\eta$ to obtain an acceptable privacy region and improve utility). We mainly focus on $f$-divergences as utility functions, and study them through the lens of contraction coefficients as done in~\cite{asoodeh:22}. 

\subsection{Preliminaries}

Given a convex function $f: [0, \infty) \rightarrow \mathbb{R}$ satisfying $f(1) = 0$, the $f$-divergence between two distributions $P_0$ and $P_1$ (over a common alphabet $\X$) is defined by 
\begin{align}
\label{eq:def-fdiv}
    D_f(P_0||P_1) =\E_Q \left[ \frac{dP_0}{dP_1} \right],
\end{align}
where $\frac{dP_0}{dP_1}$ is the Radon-Nikodym derivative. The family of $f$-divergences includes commonly used divergences such as: total variation distance (by choosing $f(t) = \max\{0,t-1\}$), KL divergence ($f(t) = t\log t$), and $\chi^2$-divergence ($f(t) = (t-1)^2$). It is well known that $f$-divergences satisfy the data processing inequality. That is, given $P_0$, $P_1$ and a channel $Q_{Y|X}$, let $M_0$ and $M_1$ be the induced marginal distributions over $\Y$ respectively, i.e., for all $S \subseteq \Y$,
\begin{align}
    M_{\nu}(S) \equiv \sum_{x \in \mathcal{X}} Q(S|x)P_{\nu}(x),
\end{align}
where $\nu \in \{0,1\}.$ Then, the data processing inequality states that 
\begin{align} \label{eq:basic-data-processing}
D_f(M_0 || M_1) \leq D_f(P_0 || P_1).    
\end{align}
Contraction coefficients are concerned with strength the above inequality to $D_f(M_0 || M_1) \leq \eta_f(Q) D_f(P_0 || P_1) $ with $\eta_f(Q) < 1$. To that end, the contraction coefficient for $f$-divergence of a given channel $Q$ is defined as 
\begin{align}
    \label{eq:def-contraction}
    \eta_f(Q) = \sup_{P_0,P_1: D_f(P_0 || P_1) \neq 0} \frac{D_f(M_0 || M_1)}{D_f(P_0||P_1)}.
\end{align}

\subsection{Definitions and Preliminary Results}

\begin{Definition}[Local Differential Privacy] \label{def:LDP}
    For $\eps \geq 0$, a mechanism $Q$ is $\eps$-locally differentially private if 
    \begin{align}
        \sup_{S \subseteq \mathcal{Y}, x, x' \in \mathcal{X}} \frac{Q(S|x)}{Q(S|x')} \leq e^{\eps}.
    \end{align}
\end{Definition}

The total variation of a mechanism in given by:
\begin{Definition}
    \label{def:LocalTV}
    Given $\eta \in [0,1]$, a mechanism $Q_{Y|X}$ is said to be $\eta$-TV locally (or $d_{TV}(Q) \leq \eta$ for short) if 
    \begin{align} \label{eq:def-localtv}
        \max_{x, x'} \, d_{TV}\left(Q(. | X = x), Q(.| X = x')\right) = \eta.
    \end{align}
\end{Definition}
In other words, the Dobrushin (contraction) coefficient~\cite{dobrushin:56} of $Q$, typically denoted by $\eta_{TV}(Q)$, is given by $\eta$, i.e., $\eta_{TV}(Q) = \eta$.  The connection to contraction coefficients is reminiscent of the recent results on the connection between local differential privacy and the contract coefficients for $E_\gamma$ divergences~\cite{asoodeh:21,zamanlooy:23}.

Given $\eps >0$ and an $\eps$-LDP channel Q, it is known that~\cite[Remark A.1]{TheCompositionTheoremForDP}
\begin{align} \label{eq:TV-bound}
    d_{TV}(Q) \leq \frac{e^\eps+1}{e^\eps-1}.
\end{align}
Now given $\eps > 0$ and $ 0 \leq \eta \leq \frac{e^\eps-1}{e^\eps+1}$, define
\begin{align}
    \label{eq:def-Q-family}
    \Q_{\eps,\eta} = \{ Q: Q \text{ satisfies } \eps\text{-LDP and }\eta\text{-TV} \}.
\end{align}
Given a collection of distributions $P_0, P_1, \ldots, P_{M-1}$, we slightly abuse notation and use the tuple $(P_0,P_1,\ldots,P_{M-1})$ to denote the $M$-ary channel $Q: \{0,1,\ldots,M-1\} \rightarrow \Y$ defined by $Q(y|i)= P_i(y)$, for all $i \in \{0,1,\ldots,M-1\}$ and $y \in \Y$.

The distributions defined in equation~\eqref{eq:def-P0-P1} will play a special role, and will be denoted by $P_0^\star$ and $P_1^\star$. That is,
\begin{align}
    \label{eq:def-P0-P1-star}
    \begin{bmatrix}
        P_0^\star \\
        P_1^\star
    \end{bmatrix} =
    \begin{bmatrix}
        \eta \frac{e^\eps}{e^\eps-1} & \eta \frac{1}{e^\eps-1} & 1 - \eta \frac{e^\eps+1}{e^\eps-1} \\
        \eta \frac{1}{e^\eps-1} & \eta \frac{e^\eps}{e^\eps-1} & 1 - \eta \frac{e^\eps+1}{e^\eps-1}
    \end{bmatrix},
\end{align}
where for notational convenience, we suppressed the dependence of $P_0^\star$ and $P_1^\star$ on $\eps$ and $\eta$. The corresponding binary-input channel will be denoted by $Q^\star_{\eps,\eta}: \{0,1\} \rightarrow \{0,1,2\}$, i.e., for all $y \in \{0,1,2\}$,
\begin{align}
    \label{eq:def-Qstar}
    Q^\star_{\eps,\eta}(y|0) = P_0(y) \text{ and } Q^\star_{\eps,\eta}(y|1) = P_1(y). 
\end{align}
 $Q^\star_{\eps,\eta}$ will be referred to as the dominating mechanism.

Considering the equivalence with the Dobrushin contraction coefficient, we get the following result:
\begin{Proposition}\label{prop:TV-contraction}
    Given distributions $P_0$ and $P_1$ over a common alphabet $\Y$, $\eps \geq 0$, and $0 \leq \eta \leq \frac{e^\eps-1}{e^\eps+1}$,
    \begin{align}
        \label{eq:prop-TV-contraction}
        \sup_{Q \in \Q_{\eps,\eta}} d_{TV} (M_0, M_1) = \eta d_{TV} (P_0, P_1).
    \end{align}
    Moreover, the supremum is achieved by the \emph{binary mechanism with erasure} $Q^{(be)}: \Y \rightarrow \{0,1,2\}$ defined by
    \begin{align}
        Q^{(be)}(.|y) = Q^\star_{\eps,\eta} \left(.| \mathbf{1} \{P_1(y) > P_0(y) \} \right),
    \end{align}
    where $\mathbf{1}\{\}$ is the indicator function and $Q^\star_{\eps,\eta}$ is the dominating mechanism defined in equation~\eqref{eq:def-Qstar}.
\end{Proposition}
This recovers and generalizes a result by Kairouz \emph{et al.}~\cite[Theorem 6, Corollary 11]{ExtremalMechanisms}, which performs the optimization over all $\eps$-LDP channels, i.e., it corresponds to $\eta = \frac{e^\eps-1}{e^\eps+1}$ (proof deferred to Appendix~\ref{app:prop-TV-contraction}).

The following result will be very useful in the sequel.
\begin{Theorem}
    \label{lem:max-fdiv}
    Given a convex function $f: [0,\infty) \rightarrow \mathbb{R}_+$ satisfying $f(1)=0$, $\eps \geq 0$, and $0 \leq \eta \leq \frac{e^\eps-1}{e^\eps+1}$, 
    \begin{align}
        \label{eq:lem-max-fdiv}
        \sup_{ \substack{(P_0,P_1) \in \Q_{\eps,\eta}} } \!\!\!\!\!\! D_f(P_0 || P_1) & = \! D_f( P_0^\star || P_1^\star) \notag \\
        & = \! \frac{\eta \left(  f \left(e^\eps \right) + e^\eps f \left( e^{-\eps} \right) \right)}{e^\eps-1} .
    \end{align}
\end{Theorem}
Denoting by $\Q_{\eps}$ the set of all $\eps$-LDP mechanisms, one obtains as an immediate corollary that
\begin{align*}
    \sup_{ \substack{(P_0,P_1) \in \Q_{\eps}} } D_f(P_0 || P_1) = \frac{1}{e^\eps+1} \left(  f \left(e^\eps \right) + e^\eps f \left( e^{-\eps} \right) \right).
\end{align*}

\emph{Proof:}  
    Consider two distributions $P_0$ and $P_1$ over a common alphabet $\Y$ such that $(P_0,P_1) \in \Q_{\eps,\eta}$. Let $f(P_0,P_1)$ be the corresponding ROC curve (cf. Definition~\ref{def:ROC}), Then, by Corollary~\ref{corr:DP-TV-HP}, $f(P_0, P_1) \geq f(P_0^\star,P_1^\star)$ since $(P_0^\star, P_1^\star)$ achieves equality in Corollary~\ref{corr:DP-TV-HP}. Hence, it follows from Blackwell's theorem~\cite{BlackwellEquivalentExperiments} that there exists a channel $Q: \{0,1,2\} \rightarrow \Y$ such that $P_i = Q \circ P_i^\star$, $i \in \{0,1\}$. Thus it follows from the data processing inequality for $f$-divergences that
    \begin{align*}
        D_f(P_0 || P_1) = D_f( Q \circ P_0^\star || Q \circ P_1^\star) \leq D_f(P_0^\star || P_1^\star).
    \end{align*}

\subsection{Contraction Coefficient for KL Divergence}

Given two distributions $P_0$ and $P_1$ over the same alphabet $\Y$, the KL divergence is given $D(P_0 || P_1) = \E \left[ \log \frac{dP_0}{dP_1} \right]$. We provide an exact characterization of the maximum contraction coefficient $\eta_{KL}$ for channels $Q \in \Q_{\eps,\eta}$.

\begin{Theorem} \label{thm:contraction-KL}
Given $\eps \geq 0$ and $0 \leq \eta \leq \frac{e^\eps-1}{e^\eps+1}$,
\begin{align}
    \label{eq:thm-contraction-KL}
    \sup_{Q \in \Q_{\eps,\eta}} \eta_{KL} (Q) = \eta_{KL} \left( Q^\star_{\eps,\eta} \right) = \eta \frac{e^\eps-1}{e^\eps+1}.
\end{align}
\end{Theorem}
Maximizing $\eta_{KL}(Q)$ for $\eps$-LDP channels corresponds to setting $\eta$ to its maximum value, i.e., $\eta = \frac{e^\eps-1}{e^\eps+1}$ and thus
\begin{align}
    \sup_{Q \in \Q_\eps} \eta_{KL}(Q) = \left( \frac{e^\eps-1}{e^\eps+1} \right)^2,
\end{align}
which recovers Theorem 1 of~\cite{asoodeh:22}. Moreover, consider the $M$-ary randomized response mechanism defined by $Q_R: \{1,2,\ldots,M\} \rightarrow \{1,2,\ldots,M\}$ satisfying
\begin{align}
    Q_R(j|i) = \begin{cases}
        \frac{e^\eps}{e^\eps+M-1}, & i=j \\
        \frac{1}{e^\eps+M-1}, & i \neq j.
    \end{cases}
\end{align}
Then, by Theorem~\ref{thm:contraction-KL} and~\cite[Theorem 1]{choi:94}, we get
\begin{align}
    \eta_{\chi^2} (Q_R) = \eta_{KL} (Q_R) \leq \frac{(e^\eps-1)^2}{(e^\eps+M-1)(e^\eps+1)}. 
\end{align}
By Proposition 1 of~\cite{asoodeh:22},  the upper bound is tight, thus highlighting the refined analysis gained by incorporating TV in our analysis. Indeed, the proof of Theorem~\ref{thm:contraction-KL} generalizes (and in fact simplifies) the proof of Asoodeh and Zhang~\cite[Theorem 1]{asoodeh:22} by utilizing Lemma~\ref{lem:max-fdiv}.

Moreover, one can make use of Theorem~\ref{thm:contraction-KL} for a more refined comparison of two privatization mechanisms. For instance, suppose we construct two feasible (according to certain privacy constraints) privatization mechanisms $Q_1$ and $Q_2$. If $Q_1$ is stochastically degraded from $Q_2$ (i.e., $Q_1$ can be simulated from $Q_2$), then for any utility function that satisfies the data processing inequality (e.g., an $f$-divergence), $Q_2$ is a better choice (and vice versa if $Q_2$ is degraded from $Q_1$). If there is no degradedness relationship between the two, then how do we choose? In this case, one may use $\eta_Q \frac{e^{\eps_Q}-1}{e^{\eps_Q}+1}$ as a simple (and generally easily to compute) proxy for the utility function.

\begin{proof}[Proof of Theorem~\ref{thm:contraction-KL}]
    By Theorem 1 of~\cite{ordentlich:21}, it suffices to consider binary-input channels $Q = (P_0, P_1)$. Now by (the proof of) Theorem 21 in~\cite{polyanskiy:17}, we have
    \begin{align}
        \eta_{KL}((P_0,P_1))  = \sup_{\beta \in (0,1)} LC_{\beta} (P_0 || P_1),
    \end{align}
    where $LC_{\beta}(P_0 || P_1)$ is an $f$-divergence, known as the Le Cam divergence, and is given by~\cite{lecam:12,gyorfi:2001,raginsky:16}
    \begin{align}
        LC_{\beta} (P_0 || P_1) = \beta (1-\beta) \E_Q \left[ \frac{(\frac{dP}{dQ}-1)^2}{\beta \frac{dP}{dQ} + 1-\beta } \right].
    \end{align}
    Therefore,
    \begin{align}
        \sup_{Q \in \Q_{\eps,\delta}}  \!\! \eta_{KL}(Q) & = \!\!  \sup_{(P_0,P_1)  \in \Q_{\eps,\delta}} \eta_{KL}((P_0,P_1)) \notag \\
        & = \!\!\! \sup_{(P_0,P_1)  \in \Q_{\eps,\delta}} \sup_{\beta \in (0,1)} \! LC_{\beta} (P_0 || P_1)  \notag  \\
        & = \!\! \sup_{\beta \in (0,1)} \sup_{(P_0,P_1)  \in \Q_{\eps,\delta}}  \!\! LC_{\beta} (P_0 || P_1) \notag  \\
        & \stackrel{\text{(a)}} = \sup_{\beta \in (0,1)} LC_{\beta}(P_0^\star || P_1^\star)  \notag \\
        & = \eta (e^\eps-1) \sup_{\beta \in (0,1)} \beta(1-\beta) \cdot  \notag \\
        & \quad \left( \! \frac{1}{\beta e^\eps \! + \! 1 \! - \! \beta} \! + \! \frac{1}{(1 \! - \! \beta) e^\eps+\! \beta} \right)  \notag \\
        & \stackrel{\text{(b)}} = \eta \frac{e^\eps-1}{e^\eps+1},  
    \end{align}
    where (a) follows from Lemma~\ref{lem:max-fdiv}, and (b) follows from elementary algebraic manipulations (the supremum is achieved for $\beta=1/2$).
\end{proof}

\subsection{Contraction in terms of TV}

A well-known feature of $\eps$-locally DP mechanisms is that, even if the input KL divergence $D(P_0 || P_1)$ is infinite, the induced output KL divergence $D(M_0 || M_1)$ is finite. This is captured by bounding $D(M_0 || M_1)$ in terms of $d_{TV}(P_0 || P_1)$ (which is finite by definition). In particular, we have $D(M_0 || M_1) \leq \log(1+\chi^2(M_0 ||M_1))$~\cite{sason:16} and 
\begin{Theorem} \label{thm:chi2-bound}
    Given $\eps \geq 0$, $0 \leq \eta \leq \frac{e^\eps+1}{e^\eps-1}$, and $Q: \X \to \Y$ satisfying $Q \in \Q_{\eps,\eta}$, then for any distributions $P_0$ and $P_1$ on $\X$, the induced marginals $M_0$ and $M_1$ satisfy
    \begin{align} \label{eq:thm-chi2-bound}
        \chi^2 (M_0 || M_1) \leq 4 \eta (e^\eps-1)(e^{-\eps}+1)d^2_{TV}(P_0,P_1).
    \end{align}
\end{Theorem}

Considering the maximum value of $\eta = \frac{e^\eps+1}{e^\eps-1}$, this recovers Theorem 2 of~\cite{asoodeh:22}. Indeed, once again the proof refines and simplifies the proof of~\cite{asoodeh:22}. 

\begin{proof}
    As done by Asoodeh and Zhang~\cite{asoodeh:22}, we first utilize~\cite[Proposition 8]{duchi:18}:
    \begin{align*}
        \chi^2(M_0 || M_1) \leq 4 d^2_{TV}(P_0,P_1) \! \max_{x,x'} \chi^2( Q(.|x) || Q(.|x') ) .
    \end{align*}
    Now note that it follows straightforwardly from Definitions~\ref{def:LDP} and~\ref{def:LocalTV} that $Q \in \Q_{\eps,\eta} \Rightarrow ( Q(.|x) || Q(.|x') ) \in \Q_{\eps,\eta}$ for all $(x,x') \in \X^2$. Hence, it follows from Lemma~\ref{lem:max-fdiv} that
    \begin{align*}
        \chi^2(M_0 || M_1) & \leq 4 d^2_{TV}(P_0, P_1)  \max_{(x,x') \in \X^2} \chi^2 (P_0^\star || P_1^\star) \\
        & = 4 \eta (e^\eps-1)(e^{-\eps}+1)d^2_{TV}(P_0, P_1),
    \end{align*}
    where the equality is obtained by plugging $f(t) = (t-1)^2$ in Lemma~\ref{lem:max-fdiv} (and some basic algebraic manipulations).
\end{proof}

\subsubsection{Performance of the Binary Mechanism with Erasure}
Fix $P_0$ and $P_1$. Let $Q^\star \in \Q_{\eps,\delta}$ be the mechanism maximizing $D(M_0 || M_1)$ and denote the induced marginals by $M_0^\star$ and $M_1^\star$. Now consider the binary mechanism with erasure $Q^{(be)}$ (cf. Proposition~\ref{prop:TV-contraction}) and let $M_0^{(be)}$ and $M_1^{(be)}$ be the induced marginals.
Then, by Theorem~\ref{thm:chi2-bound},
\begin{align*}
    & 4 \eta (e^\eps-1)(e^{-\eps}+1)d^2_{TV}(P_0, P_1)  \geq \chi^2 \left(M_0^\star ||M_1^\star\right) \\
    & \geq  D \left(M_0^\star ||M_1^\star\right) 
     \geq D \left( M_0^{(be)} || M_1^{(be)} \right) \\
     & 
     \stackrel{\text{(a)}} \geq 2 d_{TV}^2 \left(M_0^{(be)} || M_1^{(be)} \right) 
      \stackrel{\text{(b)}} = 2 \eta^2 d_{TV}^2 \left(P_0 || P_1 \right), 
\end{align*}
where (a) follows from Pinsker's inequality and (b) follows from Proposition~\ref{prop:TV-contraction}. As such,
\begin{align} \label{eq:ratio-bin-opt}
    \frac{D \left( M_0^{(be)} || M_1^{(be)} \right)}{D \left(M_0^\star ||M_1^\star\right)} \geq \frac{\eta}{2 (e^\eps-1)(e^{-\eps}+1)}.
\end{align}
Consider the high-privacy regime where $\eps \ll 1$, and suppose $\eta \geq c \eps$ for some $c \in (0,1/2)$. Then the right-hand side of~\eqref{eq:ratio-bin-opt} goes to $\frac{c}{4}$ as $\eps$ goes to 0. That is, for the high-privacy regime, the binary mechanism with erasure is order-optimal. It is worth noting that the binary mechanism in the local DP setting (where no TV constraint is imposed) is optimal for all $\eps$'s below some threshold $\eps^\star$~\cite[Theorem 5]{ExtremalMechanisms}.

\subsection{Conversion from Local Differential Privacy}
Consider the results in Proposition~\ref{prop:TV-contraction}, and Theorems~\ref{thm:contraction-KL} and~\ref{thm:chi2-bound}. By comparison with the pure local DP counterparts, we notice that in all three cases, the relationship is a multiplicative factor equal to $\eta \frac{e^\eps+1}{e^\eps-1}$. We show that an inequality of this form hold more generally.

In particular, fix two distributions $P_0$ and $P_1$ on $\X$  and a convex function $f$ satisfying $f(1)=0$. Let,
\begin{align*}
    \opt_{\eps,\eta} = \max_{Q \in \Q_{\eps,\eta}} ~D_f (M_0 || M_1), 
   \end{align*} 
 and
 \begin{align*}
    \opt_{\eps} =  \max_{Q \in Q_\eps}  ~D_f (M_0 || M_1). 
\end{align*}
Clearly, $\opt_\eps \geq \opt_{\eps,\eta}$. Our next result proves a reverse inequality:
\begin{Theorem}
    \label{thm:opt-opt}
    For any $f$-divergence (where $f$ is convex and satisfies $f(1)=0$), $\eps \geq 0$, and $\eta \in [0, \frac{e^\eps-1}{e^\eps+1}]$,
    \begin{align}
        \label{eq:thm-opt-opt}
        \opt_{\eps,\eta} \geq \eta \frac{e^\eps+1}{e^\eps-1} \opt_{\eps}.
    \end{align}
\end{Theorem}
The proof relies on the following observation: given a channel $Q \in \Q_{\eps}$, then composing $Q$ with an erasure channel $W$ with parameter $\alpha$ yields a channel $W \circ Q \in \Q_{\eps, (1-\alpha) \frac{e^\eps-1}{e^\eps+1}}$. The details are deferred to Appendix~\ref{app:thm-opt-opt}.

\bibliographystyle{IEEEtran}
\bibliography{IEEEabrv,ref}

\newpage

\appendix

\subsection{Proof for Approximate Differential Privacy} \label{app:general-delta}

Generalizing the result to approximate differential privacy requires the following (dominating) mechanism, which outputs an integer sampled from $P_0$ when the true database is $D_0$ (null hypothesis) and from $P_1$ when the true database is $D_1$ (alternative hypothesis). For $\varepsilon \geq 0$, $\delta>0$, and $\alpha \in [0,1]$, define:
\begin{align}
    P_0 = \begin{cases}
    \delta & \text{for } x=0,\\
    \frac{(1-\delta)(1-\alpha)e^{\varepsilon}}{(1+e^{\varepsilon})} & \text{for } x=1,\\
    \alpha(1-\delta) & \text{for } x=2,\\
    \frac{(1-\delta)(1-\alpha)}{(1+e^{\varepsilon})} & \text{for } x=3,\\
    0 & \text{for } x=4,
    \end{cases}
    \label{eq:def-P0-approximate}
\end{align}
and
\begin{align}
    P_1 = \begin{cases}
    0 & \text{for } x=0,\\
    \frac{(1-\delta)(1-\alpha)}{(1+e^{\varepsilon})} & \text{for } x=1,\\
    \alpha(1-\delta) & \text{for } x=2,\\
    \frac{(1-\delta)(1-\alpha)e^{\varepsilon}}{(1+e^{\varepsilon})} & \text{for } x=3,\\
    \delta & \text{for } x=4.
    \end{cases}
    \label{eq:def-P1-approximate}
\end{align}

The total variation of the above mechanism is $d_{TV} = \delta + \frac{(1-\delta)(1-\alpha)(e^{\varepsilon}-1)}{1 + e^{\varepsilon}}$.
Setting $\alpha = 0$ yields the same mechanism proposed in \cite{TheCompositionTheoremForDP}.

The proof is analogous to case of $\delta = 0$ in Section \ref{composed_dominating_mechanism_epsilon_0}. To see how the expression $\max_{\substack{S \in \mathcal{X}^k}} \{\tilde{P_0}(S) - e^{\varepsilon_j} \tilde{P_1}(S)\} = 1 - (1-\delta)^{k}(1-\delta_j)$ in Theorem~\ref{thm:CompositionTVApproximate} was derived, we split it into two parts ($p_1 + p_2$), such that $p_1 = 1 - (1-\delta)^k$ and $p_2 = (1-\delta)^k\delta_{j}$.\\
We are interested in the sequences $s_k$ that achieve the inequality $\tilde{P_0}(s_k) - e^{\tilde{\varepsilon}} \tilde{P_1}(s_k)>0$:
\begin{itemize}
    \item $p_1$: This inequality is satisfied for all output sequences in which $x = 0$ appears at least once, because the probability of such a sequence under $P_1$ is always 0. The probability of obtaining a sequence with no ``$0$''s under $P_0$ is $(1-\delta)^k$. Taking the complement to obtain the probability that a sequence contains at least one ``$0$'' under $P_0$, we get $p_1 = 1 - (1-\delta)^k$.
    \item $p_2$: We now consider the sequences that do not contain a ``$0$'' and that lead to a positive $\tilde{P_0}(s_k) - e^{\tilde{\varepsilon}} \tilde{P_1}(s_k)$. We apply the same logic that we followed in the case of $\delta = 0$. Factoring out $(1-\delta)^k$ from the result, we obtain $p_2 = (1-\delta)^k \delta_j$.
\end{itemize}

\subsection{Proof of Theorem~\ref{thm:central-limit}} \label{app:central-limit}

The theorem is an application of Dong \emph{et al.}'s result, in particular Theorem 3.5 in ``Gaussian Differential Privacy''~\cite{GaussianDP}, specialized for $\eps$-DP and $\eta$-TV. Note that, by definition, $M_{ni}$ is $f_{\eps_{ni},0,\eta_{ni}}$-DP. Now, following the definitions in~\cite{GaussianDP},
\begin{align*}
    kl(f_{\eps,0,\eta}) & \defeq - \int_0^1 \log \left| f_{\eps,0,\eta}'(t) \right| dt  \notag \\
    & = - \left( \int_0^{ \frac{\eta}{e^\eps-1}} \eps dt + \int_{1- \eta \frac{e^\eps}{e^\eps+1}}^1 (-\eps) dt \right) \notag \\
    & = \eps \eta. 
\end{align*}
Similarly,
\begin{align*}
    \kappa_2(f_{\eps,0,\eta})  \defeq  \int_0^1 \log^2 \left| f_{\eps,0,\eta}'(t) \right| dt  = \eta \eps^2 \frac{e^\eps+1}{e^\eps-1}, 
\end{align*}
and
\begin{align*}
    \kappa_3 (f_{\eps,0,\eta})  \defeq  \int_0^1 \left| \log \left| f_{\eps,0,\eta}'(t) \right| \right|^3 dt = \eta \eps^3 \frac{e^\eps+1}{e^\eps-1}. 
\end{align*}
Now,
\begin{align} \label{eq:kl-limit}
    \sum_{i=1}^n kl(f_{\eps_{ni},0,\eta_{ni}}) = \sum_{i=1}^n \eps_{ni} \eta_{ni} \rightarrow \frac{\mu^2}{2},
\end{align}
where the limit holds by assumption of the theorem. Similarly,
\begin{align} \label{eq:max-kl-limit}
    \max_{1 \leq i \leq n} \! kl(f_{\eps_{ni},0,\eta_{ni}}) \! = \!\!  \max_{i} \eta_{ni} \eps_{ni} \! \leq \! \max_{i} \eps_{ni} \rightarrow 0,
\end{align}
where the limit holds by assumption. Since $kl(f) \geq 0$, then $\max_{1 \leq i \leq n} kl(f_{\eps_{ni},0,\eta_{ni}}) \rightarrow 0$. Moreover, note that
\begin{align*}
   & \sum_{i=1}^n \left|\frac{1}{2} \kappa_2(f_{\eps_{ni},0,\eta_{ni}}) - kl(f_{\eps_{ni},0,\eta_{ni}}) \right|  \\
   & =
   \sum_{i=1}^n \eps_{ni} \eta_{ni} \left|  \frac{\eps_{ni}(e^\eps_{ni}+1)}{2(e^\eps_{ni}-1)} -1 \right| \\
   & \leq \sum_{i=1}^n \eps_{ni} \eta_{ni} \eps_{ni}^2 \\
   & \leq \max_{1 \leq i \leq n} \eps_{ni}^2 \sum_{i=1}^n \eps_{ni} \eta_{ni}  \rightarrow 0,
\end{align*}
where the first inequality follows from the Taylor expansion of $t(e^t+1)/(e^t-1)$. Hence,
\begin{align} \label{eq:kappa2-limit}
    \sum_{i=1}^n \kappa_2(f_{\eps_{ni},0,\eta_{ni}}) \rightarrow \mu^2.
\end{align}
  Similarly, one can show that $\sum_{i=1}^n \kappa_3(f_{\eps_{ni},0,\eta_{ni}}) \rightarrow 0$. Coupled with equations~\eqref{eq:kl-limit},~\eqref{eq:max-kl-limit}, and~\eqref{eq:kappa2-limit}, the assumptions of Theorem 3.5 of~\cite{GaussianDP} are satisfied, hence the limit of the composition converges uniformly to $G_\mu$.

\subsection{Proof of Proposition~\ref{prop:subsampling}} \label{app:subsampling}
    Since $M$ is $(\eps,\delta)$-DP then $M \circ \Sample_{m}$ is ($\log \left( 1+ p \left(e^\eps -1\right) \right), p \delta$)-DP by Theorem 29 of~\cite{steinke2022composition}. Similarly, $M$ being $\eta$-TV can be rewritten as $(0,\eta)$-DP (see, for instance, equations of Corollary~\ref{corr:DP-TV-HP}). Hence, $M \circ \Sample_{m}$ is ($0$,$p\eta$)-DP, i.e., it is $p\eta$-TV, as desired. Although for $(\eps,\delta)$-pair, the result is known to be tight, it is not clear a priori whether the result is tight \emph{simultaneously} for both constraints. We show that is tight by exhibiting a mechanism that achieves the region. Consider two neighboring databases $D_0$ and $D_1$ that differ on one entry. Say there exists an element  $a \in D_0$ but $a \notin D_1$. Consider the following mechanism:
    \begin{align*}
        M(D)  \sim P_0, & \text{ if } a \in D \\
        M(D)  \sim P_1, & \text{ if } a \notin D,
    \end{align*} 
    where $P_0$ and $P_1$ are defined in equation~\eqref{eq:def-P0-P1}. Then, after subsampling with ratio $p$, the induced distributions will be as follows:
\begin{align}
        M(D)  \sim pP_0+(1-p)P_1, & \text{ if } a \in D \\
        M(D)  \sim P_1, & \text{ if } a \notin D.
    \end{align}
One can explicitly compute the corresponding privacy region, and obtain the desired result by considering the worst-case among the above region and the one induced (by symmetry) from the pair of distributions
\begin{align}
    pP_1 + (1-p)P_0 \text{ and } P_0.
\end{align}
    
\subsection{Total Variation of the Laplace Mechanism} \label{app:laplace}

To determine the total variation of the Laplace mechanism, consider two Laplace distributions, one centered at zero ($P: Lap(\frac{1}{\varepsilon},0)$), and one shifted to the right by the sensitivity ($Q: Lap(\frac{1}{\varepsilon},\Delta)$).  
\begin{align*}
        d_{TV}(P,Q) &= \int_{x:P(x)>Q(x)}^{} P(x) - Q(x) \,dx \\
        &= \int_{-\infty}^{\Delta/2} \frac{\eps}{2\Delta} e^{-\frac{\eps|x|}{\Delta}} - \frac{\eps}{2\Delta} e^{-\frac{\eps|x-\Delta|}{\Delta}} \,dx \\
        & = \int_{-\infty}^{0} \frac{\eps}{2\Delta}(e^{\eps x} - e^{-\eps(\Delta-x)}) \,dx  \\
        & ~~ + \int_{0}^{\frac{\Delta}{2}} \frac{\eps}{2\Delta}(e^{-\eps x} - e^{-\eps(\Delta-x)}) \,dx \\
    & = 1 - e^{-\frac{\varepsilon}{2}}.
\end{align*}

\subsection{Total Variation of the Gaussian Mechanism} \label{app:gaussian}
Corollary $2.13$ in \cite{GaussianDP} states that a mechanism is $\mu$-GDP if and only if it is ($\eps$,$\delta(\eps)$)-DP for all $\eps \geq 0$, where $\delta(\eps) = \Phi\bigl( - \frac{\eps}{\mu} + \frac{\mu}{2}\bigr) - e^{\eps}\Phi\bigl( - \frac{\eps}{\mu} - \frac{\mu}{2}\bigr)$. Thus, the total variation distance of the Gaussian mechanism is 
\begin{align}
    \delta(0) = \Phi\bigl(\frac{\mu}{2}\bigr) - \Phi\bigl(- \frac{\mu}{2}\bigr) = 2\Phi\bigl(\frac{\mu}{2}\bigr) - 1.
\end{align}

\subsection{Analysis of the Staircase Mechanism}
\subsubsection{Total Variation}
\label{app:staircase}

To derive the total variation of the staircase mechanism, we consider two distributions: $P(x)$ centered at $0$, and $Q(x)$ which is $P(x)$ shifted to the right by the sensitivity $\Delta$. We obtain the total variation of the staircase mechanism by subtracting the areas under each distribution for the values of $x$ where $P(x)>Q(x)$. 
We distinguish the cases of $\gamma \leq \frac{1}{2}$ and $\gamma>\frac{1}{2}$. 
For $\gamma \leq \frac{1}{2}$, we are interested in the interval $(-\infty,\gamma]$ which we split into two sub-intervals (Figure~\ref{fig:staircase_dtv_1}):
\begin{itemize}
    \item For $(-\infty,-\gamma]$, the difference between the area under the red curve and that under the blue curve is given by $a(\gamma) \Delta \sum_{n = 1}^{\infty} {(e^{-n\varepsilon} - e^{-n\varepsilon-\varepsilon})} = a(\gamma)\Delta e^{-\varepsilon}$
    \item For $(-\gamma,\gamma]$, the difference is $a(\gamma)(2\gamma \Delta (1 - e^{\varepsilon}))$
\end{itemize}

\begin{figure}[htp]
\centering
\includegraphics[width=\linewidth]{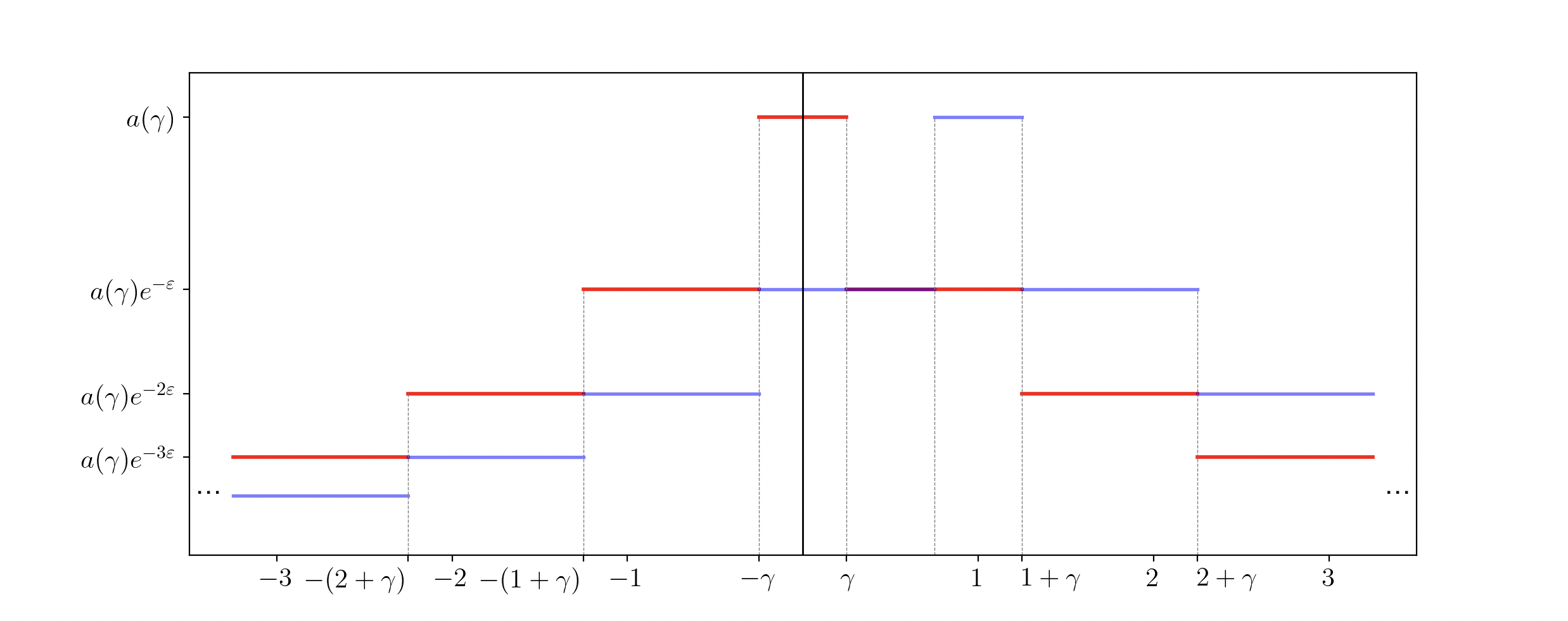}
\caption{Two staircase distributions for $\gamma\leq \frac{1}{2}$ shifted by the sensitivity.}
\label{fig:staircase_dtv_1}
\end{figure}

For $\gamma>\frac{1}{2}$, we look at the interval $(-\infty,1-\gamma)$ (Figure~\ref{fig:staircase_dtv_2}). The total variation of the mechanism is given by $a(\gamma) \Delta \sum_{n = 0}^{\infty} {(e^{-n\varepsilon} - e^{-n\varepsilon-\varepsilon})} = a(\gamma)\Delta$.

\begin{figure}[htp]
\centering
\includegraphics[width=\linewidth]{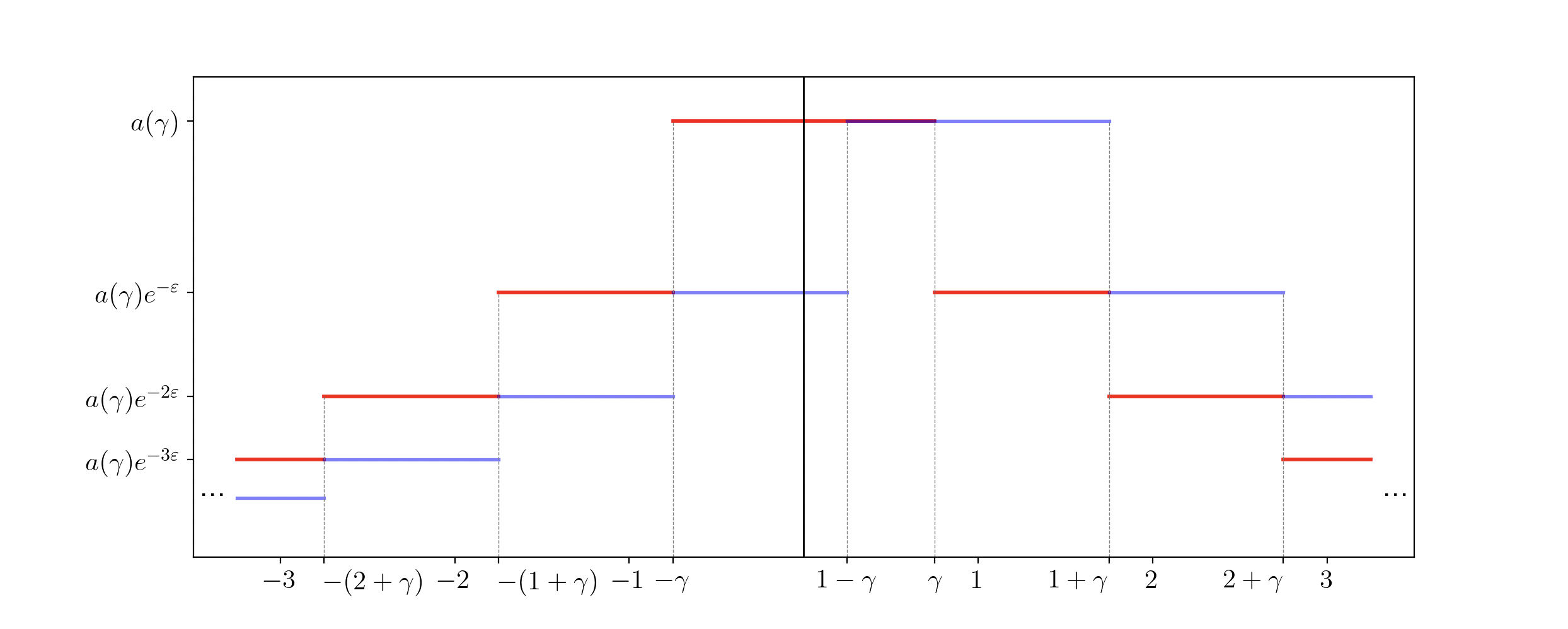}
\caption{Two staircase distributions for $\gamma>\frac{1}{2}$ shifted by the sensitivity.}
\label{fig:staircase_dtv_2}
\end{figure}

Considering two intervals for the values of $\gamma$ and setting $\Delta = 1$, we derive the following expressions for the total variation of the staircase mechanism
\begin{equation*} d_{TV} \!\! = \!\!
  \left\{
    \begin{aligned}
        & \frac{(1 - e^{-\eps}) (2 \gamma (1-e^{-\varepsilon}) + e^{-\varepsilon})}{2(\gamma + e^{-\varepsilon}(1 - \gamma))}, &\gamma \in [0,\frac{1}{2}), \\
        & \frac{1 - e^{-\varepsilon}}{2(\gamma + e^{-\varepsilon}(1 - \gamma))}, & \gamma \in [\frac{1}{2}, \infty).
    \end{aligned}
  \right.
\end{equation*}

\subsubsection{Privacy Region} \label{app:staircase-region}

The staircase mechanism with total variation $\eta$ achieves exactly the privacy region corresponding to $(\eps,0)$-DP $\eta$-TV (cf. Corollary~\ref{corr:DP-TV-HP}).
Consider two staircase-shaped distributions, $P_0$ and $P_1$ (where $P_1$ is $P_0$ shifted by $\Delta$) for some arbitrary $\gamma$. The Neyman-Pearson Lemma states that the optimal decision rule is to reject when the likelihood ratio is above some threshold $T$. Since the likelihood ratio of $P_1$ to $P_0$ is non-decreasing (cf. Figures~\ref{fig:staircase_dtv_1} and~\ref{fig:staircase_dtv_2} in Appendix~\ref{app:staircase}), we can pick a point $t$ and accept whenever $t>x$. The rejection rules $x\geq 1-\gamma$ and $x \geq \gamma$ will yield type I and II values that correspond to the points ($\frac{\eta}{e^{\eps}-1}$,$1-\frac{\eta e^{\eps}}{e^{\eps}-1}$) and ($1-\frac{\eta e^{\eps}}{e^{\eps}-1}$, $\frac{\eta}{e^{\eps}-1}$) in Figure~\ref{fig:roc_dp_tv} for $\delta=0$. The rest of the tradeoff function is derived by averaging. Therefore, the bound provided in Theorem~\ref{thm:CompositionTVApproximate} exactly describes the privacy region of the composed staircase mechanism.

\subsection{Proof of Proposition~\ref{prop:TV-contraction}} \label{app:prop-TV-contraction}    

Recall Dobrushin's classical result~\cite{dobrushin:56}:
\begin{align}
    \label{DobrushinResult}
   \! \sup_{\substack{P_0,P_1
    \\P_0 \neq P_1}}  \!\!\! \frac{||M_0 - M_1||_{\text{TV}}}{||P_0 - P_1||_{\text{TV}}}\! = \! \sup_{x, x'}  d_{TV}\bigl(Q(.|  x), Q(.|  x')\bigr). 
\end{align}
    As such,
    \begin{align*}
       \|M_0 - M_1 \|_{TV} \leq \eta_{TV}(Q)  \|P_0 - P_1 \|_{TV} = \eta \|P_0 - P_1 \|,
    \end{align*}
    where the inequality follows from~\eqref{DobrushinResult}, and the equality follows from Definition~\ref{def:LocalTV}.
    It remains to show that the binary mechanism with erasure achieves
\begin{align}
    ||M_0 - M_1||_{\text{TV}} = \eta \, ||P_0 - P_1||_{\text{TV}}.
\end{align}
Let $\mathcal{A} \subseteq \mathcal{X}$ such that $\mathcal{A}= \{x \in \X: P_0(x) \geq P_1(x)\}$. Hence
\begin{align}
    \|P_0-P_1\|_{TV} = P_0(A)-Q_0(A).
\end{align}
Now note that
\begin{align*}
    M_\nu(0) & = \sum_{x\in \mathcal{A}} \eta \frac{ e^\eps}{e^\eps-1} P_{\nu}(x) + \sum_{x\notin \mathcal{A}}  \eta \frac{ 1}{e^\eps-1} P_{\nu}(x), \\
    M_\nu(1) & = \sum_{x\in \mathcal{A}} \eta \frac{ 1}{e^\eps-1} P_{\nu}(x) + \sum_{x\notin \mathcal{A}} \eta \frac{ e^\eps}{e^\eps-1} P_{\nu}(x), \\
    M_{\nu}(e) & = \left(1 - \eta \frac{1 + e^{\varepsilon}}{e^{\varepsilon}-1} \right) \sum_{x}P_\nu(x) = \left(1 - \eta \frac{1 + e^{\varepsilon}}{e^{\varepsilon}-1} \right).
\end{align*}
Therefore,
\begin{align*}
    & |M_0(0) - M_1(0)| \\
    & = \sum_{x\in \mathcal{A}} \eta \frac{ e^\eps}{e^\eps-1} P_{0}(x) + \sum_{x\notin \mathcal{A}} \eta \frac{1}{e^\eps-1} P_{0}(x) \\
    & \quad - \sum_{x\in \mathcal{A}} \eta \frac{ e^\eps}{e^\eps-1} P_{1}(x) - \sum_{x\notin \mathcal{A}}  \eta \frac{ 1}{e^\eps-1} P_{1}(x)\\
    & = \eta \frac{ e^\eps}{e^\eps-1} \sum_{x\in \mathcal{A}}(P_0(x) - P_1(x)) \\
    & \quad + \eta \frac{ 1}{e^\eps-1} \sum_{x \notin \mathcal{A}}(P_0(x) - P_1(x)) \\
    & = \eta \frac{ e^\eps}{e^\eps-1} ||P_0 - P_1||_{TV} - \eta \frac{1}{e^\eps-1}||P_0 - P_1||_{TV} \\
    & = \eta  ||P_0 - P_1||_{TV} 
\end{align*}
Since $|M_0(0) - M_1(0)| = |M_0(1) - M_1(1)|$, and $|M_0(e) - M_1(e)| = 0$, we get
\begin{align*}
    \|M_0 - M_1 \|_{TV}  = |M_0(0) - M_1(0)| =  \eta \|P_0 - P_1 \|_{TV}.
\end{align*}

\subsection{Proof of Theorem~\ref{thm:opt-opt}} \label{app:thm-opt-opt}

    Let $Q_{Y|X}$ be $\eps$-LDP. Let $W_\alpha: \Y \to \Y \cup \{e\}$ be an erasure channel with parameter $\alpha$, i.e., for all $y \in \Y$, $W(y|y)  = 1-\alpha$ and $W(e|y) = \alpha$. 

    Let $Q'_{Y|X} = W \circ Q_{Y|X}$. Note that for $y \in \Y$,
    $Q'(y|x) = Q(y|x) (1-\alpha)$ for all $x \in \X$, and $Q(e|x) = \sum_{y \in \Y} Q(y|x) W(e|y) = \alpha$ for all $x \in \X$.    
    \begin{Lemma}
        $d_{TV}(Q'_{Y|X}) \leq (1-\alpha) \frac{e^\eps-1}{e^\eps+1}$. 
    \end{Lemma}
    \begin{proof}
        Fix $x, x'$, and $A \subset \Y \cup \{e\}$. Then,
        \begin{align*}
            Q'(A|x) \! - \! Q'(A|x') \! &  = \! \sum_{y' \in A} Q'(y|x)-Q'(y|x') \\
            & = \!\!\! \sum_{y' \in A \setminus \{e\}} \!\!\! (1-\alpha)(Q(y|x)-Q(y|x')) \\
            & \leq (1-\alpha) \| Q(.|x) - Q(.|x') \|_{TV}.
        \end{align*}
        Taking supremum over all $x$, $x'$, and $A$, and noting that $d_{TV}(Q) \leq \frac{e^\eps-1}{e^\eps+1}$ (by equation~\eqref{eq:TV-bound}), we get the desired result.
    \end{proof}
    As such, choosing $\alpha  =1- \eta \frac{e^\eps+1}{e^\eps-1}$, we get $Q'$ satisfying $d_{TV}(Q') \leq \eta$. It is straightforward to check that $Q'$ also satisfies $\eps$-LDP so that $Q'$ is feasible for the $\opt_{\eps,\eta}$ problem. Let $M_0'$ and $M_1'$ be the induced marginals. Then, 
    \begin{align}
        M_0'(e) = \sum_{x \in \X } P_0(x) Q'(e|x) = \alpha,
    \end{align}
    and for all $y \in \Y$,
    \begin{align*}
        M_0'(y) & =  \sum_{x \in \X } P_0(x) Q'(y|x) \\
        & =  \sum_{x \in \X } P_0(x) (1-\alpha) Q(y|x) = (1-\alpha) M_0(y). 
    \end{align*}
    Analogous equations hold for $M_1'.$ Hence,
    \begin{align*}
        & D_f(M_0'||M_1') \\
        & = \sum_{y' \in \Y \cup \{e\} } M_1'(y) f \left( \frac{M_0'(y)}{M_1'(y)} \right) \\
        & = \sum_{y \in \Y} (1-\alpha) M_1(y)  f \left( \frac{M_0(y)}{M_1(y)} \right)  + M_1(e) f \left( \frac{M_0(e)}{M_1(e)} \right) \\
        &\stackrel{\text{(a)}} = (1-\alpha) D_f(M_0||M_1) \\
        &\stackrel{\text{(b)}} = \eta \frac{e^\eps+1}{e^\eps-1} D_f(M_0||M_1),
    \end{align*}
    where (a) follows from the fact that $f(1)=0$, and (b) follows from the choice of $\alpha = 1- \eta \frac{e^\eps+1}{e^\eps-1}$.
    Finally, taking supremum on both sides yields
    \begin{align}
        \opt_{\eps,\eta} \geq \eta \frac{e^\eps+1}{e^\eps-1} \opt_{\eps}.
    \end{align}

\end{document}